\documentclass[conference]{IEEEtran}

\makeatletter
\newcommand\semihuge{\@setfontsize\semihuge{22.3}{22}}
\makeatother

\usepackage{algpseudocode}
\usepackage{algorithm}
\usepackage{algorithmicx}
\usepackage{latexsym}
\usepackage{lipsum} % For algorithm smape after or before:
\usepackage[dvips]{color}
\usepackage{comment}
\usepackage{todonotes}
\usepackage{epsf}
\usepackage{times}
\usepackage{epsfig}
\usepackage{graphicx}
\usepackage{bbold}
\usepackage{mathtools}
\usepackage{mathrsfs}
\usepackage{amssymb}
\usepackage{pdfpages}
\usepackage{epstopdf}
\usepackage{float}
\usepackage{amsfonts}
\usepackage{here}
\usepackage{rawfonts}
\usepackage{dsfont}
\usepackage{lettrine} % \lettrine[findent=1pt]{{{R}}}{}
\usepackage{url}
\usepackage{amsthm}
\usepackage{cite}
\allowdisplaybreaks
\usepackage{csquotes}
\usepackage{verbatim}
\usepackage[english]{babel}

\usepackage{amsmath}

\newfloat{algorithm}{t}{lop}

\ifCLASSOPTIONcompsoc
    \usepackage[caption=false, font=normalsize, labelfont=sf, textfont=sf]{subfig}
\else
\usepackage[caption=false, font=footnotesize]{subfig}
\fi

\newtheorem{theorem}{\bf Theorem}

\newcommand\blfootnote[1]{%
  \begingroup
  \renewcommand\thefootnote{}\footnote{#1}%
  \addtocounter{footnote}{-1}%
  \endgroup
}
	
\DeclareMathOperator*{\argmax}{\arg\!\max}
	
\begin{document}\bstctlcite{IEEEexample:BSTcontrol}
\title{\LARGE Distributed Learning for Low Latency Machine Type Communication in a Massive Internet of Things \vspace{-0.3cm}}    

\author{\IEEEauthorblockN{Taehyeun Park$^1$ and Walid Saad$^1$}\vspace{-0.05cm}\\
	\IEEEauthorblockA{\small $^1$Wireless@VT, Bradley Department of Electrical and Computer Engineering, Virginia Tech, Blacksburg, VA, USA,\\ Emails:\{taehyeun, walids\}@vt.edu}
	\vspace{-0.92cm}}
\maketitle\vspace{-0.8cm}
\vspace{-0.1cm}
\begin{abstract}
The Internet of Things (IoT) will encompass a massive number of machine type devices that must wirelessly transmit, in near real-time, a diverse set of messages sensed from their environment. Designing resource allocation schemes to support such coexistent, heterogeneous communication is hence a key IoT challenge. In particular, there is a need for self-organizing resource allocation solutions that can account for unique IoT features, such as massive scale and stringent resource constraints. In this paper, a novel \emph{finite memory multi-state sequential learning} framework is proposed to enable diverse IoT devices to share limited communication resources, while transmitting both delay-tolerant, periodic messages and urgent, critical messages. The proposed learning framework enables the IoT devices to learn the number of critical messages and to reallocate the communication resources for the periodic messages to be used for the critical messages. Furthermore, the proposed learning framework explicitly accounts for IoT device limitations in terms of memory and computational capabilities. The convergence of the proposed learning framework is proved, and the lowest expected delay that the IoT devices can achieve using this learning framework is derived. Furthermore, the effectiveness of the proposed learning algorithm in IoT networks with different delay targets, network densities, probabilities of detection, and memory sizes is analyzed in terms of the probability of a successful random access request and percentage of devices that learned correctly. Simulation results show that, for a delay threshold of $1.25$ ms, the average achieved delay is $0.71$ ms and the delay threshold is satisfied with probability $0.87$. Moreover, for a massive network, a delay threshold of $2.5$ ms is satisfied with probability $0.92$. The results also show that the proposed learning algorithm is very effective in reducing the delay of urgent, critical messages by intelligently reallocating the communication resources allocated to the delay-tolerant, periodic messages. 
\end{abstract} 

\section{Introduction} \label{sec:intro}%\vspace{-0.01cm}
\blfootnote{\noindent This research was supported by the U.S. Office of Naval Research (ONR) under Grant N00014-15-1-2709.}The Internet of Things (IoT) is an emerging networking technology that promises to interconnect a massive number of devices such as wearables, sensors, smartphones, and other machine type devices \cite{saad}. {The IoT will impact multiple application domains including home automation \cite{homeauto}, smart grids \cite{smartgrid, lat2}, drone-based systems \cite{moha1}, healthcare systems \cite{rvwr14}, and industrial monitoring \cite{lat1, generaliot}.} To support such innovative IoT applications, there is a need for new wireless technologies that can enable a large-scale connectivity among IoT devices. {However, integrating the IoT ecosystem into existing wireless networks faces many challenges that include coexistence with human-type devices, self-organizing operation, and limited communication resources \cite{limitspec, rvwr14}.} Moreover, the IoT devices are typically machine-type devices that differ significantly from conventional human-type devices, such as smartphones, in terms of performance requirements, memory, computation, and energy constraints as well as traffic patterns \cite{saad}. In addition, the IoT devices will require low latency, ultra-reliable, and short packet transmissions. Therefore, existing wireless networks must be re-designed to meet these IoT challenges.\\
\indent The IoT devices deployed in an existing wireless network will need to deliver diverse applications and, thus, they will be heterogeneous in terms of their performance requirements and traffic patterns. For example, IoT devices used for smart metering may not require ultra-low latency and will communicate periodically. In contrast, IoT devices used for industrial monitoring will require ultra-low latency and communicate sporadically. Furthermore, such heterogeneous IoT devices with different properties and requirements will have to coexist and share limited communication resources appropriately to satisfy their quality-of-service (QoS) requirements. To incorporate such heterogeneous IoT devices into existing wireless networks with limited resources, a distributed method is necessary to allocate limited resources appropriately depending on the QoS needs of the IoT devices. The IoT devices must be able to obtain their communication resources autonomously, because it is impractical to assume that they can communicate frequently with the base station, given their stringent resource constraints. Moreover, in a massive IoT, the base station will not be able to manage the communication resources of all devices in a timely manner. Therefore, resource allocation in the IoT must be distributed and must also consider the limited capabilities of the IoT devices in terms of computation and memory. A resource allocation framework satisfying the aforementioned requirements will accelerate the deployment of the IoT over existing networks with limited communication resources.
\subsection{Existing Works}
\indent {Resource allocation in the IoT has attracted significant research recently \cite{manage1, manage11, noma, rvwr11, fog1, rvwr21, rvwr22, rvwr23, cap2, lat3, lat4, cran, cendcen, tae2}.} For instance, different multiple access schemes have been proposed in order to maximize energy efficiency \cite{manage1} or throughput \cite{manage11}. Furthermore, non-orthogonal multiple access schemes for IoT networks are analyzed in \cite{noma}, in terms of channel capacity, latency, and connectivity. {In \cite{rvwr11}, the authors introduce an edge-cloud based label-less learning to offload only valuable data to the cloud and to reduce the traffic.} The authors in \cite{fog1} propose a framework based on cognitive edge computing to optimize the use of distributed cloud resources for the IoT. {In \cite{rvwr21}, the authors integrate software-defined networking to machine-to-machine communication to enable a smart energy management for various environments.} {A learning automaton to adjust the access class barring dynamically is proposed in \cite{rvwr22}, while the work in \cite{rvwr23} introduces an energy-efficient, stable, and weak Pareto optimal matching algorithm for device-to-device communication.} In \cite{cap2} and \cite{lat3}, the authors propose different methods to integrate IoT devices into existing cellular networks and share communication resources with existing devices while avoiding network congestion and minimizing system capacity in a massive network. Meanwhile, the work in \cite{lat4} introduces new communication protocols for IoT applications that require ultra-reliability and ultra-low latency. Moreover, the authors in \cite{cran} propose a resource allocation framework exploiting cloud radio access network for joint channel selection and power allocation. In \cite{cendcen}, the authors compare the performance of dynamic time division duplex with centralized and distributed resource allocation schemes in a dense IoT. {The authors in \cite{narrowband} and \cite{rvwr13} provide an overview on power consumption, security, spectrum resource management, and deployment of narrow band IoT and cognitive low-power wide-area networks, respectively.} {In \cite{tae}, we considered a system model in which there are periodic and critical messages coexisting in an IoT system and proposed a learning framework that enables the IoT devices to autonomously allocate the limited communication resources for a highly reliable transmission of the critical messages.} {However, the work in \cite{tae} assumes that there is at most one critical message at any given time, and the IoT devices learn the existence of the critical message to decide to allocate the communication resources for the critical message or not.} In \cite{our_survey}, we provided a qualitative survey that discusses the potential of learning in the IoT. A framework based on sequential learning to learn the number of urgent, critical messages and to allocate corresponding uplink resources is studied in \cite{tae2}. Although message type heterogeneity and the resource constraints of the IoT devices are considered in \cite{tae2}, this work requires a learning sequence for each critical message. Therefore, an IoT device must be part of all learning sequences to accurately learn the number of critical messages, which may not be practical for resource-constrained IoT devices. {Meanwhile, the works in \cite{lat3, lat4, rvwr22}, and \cite{cendcen} do not account for the heterogeneity in communication requirements, traffic patterns, and message types.} Moreover, most of the prior art in \cite{cap2, manage1, lat3}, and \cite{cran} relies on centralized solutions, which may be impractical for a massive IoT. {Further, even though the energy constraints of IoT devices are typically considered, such as in \cite{rvwr23} and \cite{cran}, the IoT devices also have limited capabilities in terms of memory and computation, which are ignored in existing works, \cite{manage1, manage11, cap2, cran}, and \cite{cendcen}.}

\subsection{Contributions}
\indent The main contribution of this paper is a \emph{multi-state sequential learning framework} \cite{cover} with finite memory that can enable the IoT devices to allocate their limited communication resources in a self-organizing manner while satisfying their heterogeneous latency requirements. In particular, we consider a massive IoT system in which there are IoT devices with limited capabilities transmitting either periodic or critical messages. \emph{Periodic messages}, such as meter readings, are delay-tolerant messages with predictable traffic pattern, while \emph{critical messages}, such as system failures, are randomly occurring, delay-intolerant messages. When some devices transmit critical messages, the communication resources that are normally allocated for the periodic messages must be reallocated for the critical messages for their timely transmission. We propose a novel multi-state sequential learning algorithm that allows the IoT devices to learn the number of existing critical messages and, then, to reallocate communication resources appropriately depending on the delay requirement and the number of critical messages. The proposed multi-state sequential learning is suitable for IoT devices as it only requires finite number of observations to learn in a distributed way. Then, we analyze the memory, limited observation capabilities of devices, and expected delay resulting from our proposed algorithm. We show that our proposed learning framework converges to the true state of the environment (i.e., it autonomously learns the true number of existing critical messages). We also derive the lowest expected delay of a critical message that can be achieved with our learning framework. Simulation results show the learning effectiveness of the proposed framework as well as its ability to realize low-latency transmissions across a massive IoT network.\\
\indent The rest of this paper is organized as follows. Section II introduces the system model and Section III presents the proposed learning algorithm and analyzes its properties. Section IV analyzes the simulation results, while Section V draws conclusions. \vspace{-.2cm}
\section{System Model}\label{sec:SM}

\begin{figure}[t]
	\centering
	\includegraphics[scale = 0.65]{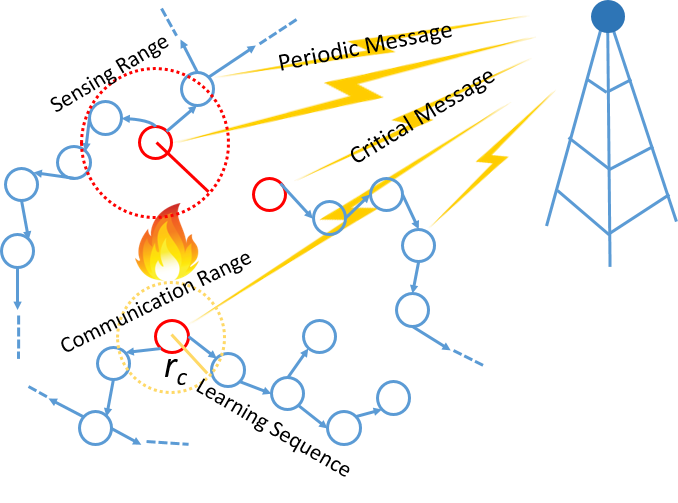}\vspace{-0.4 cm}
	\caption{{An overview of the considered problem setup showing all red devices transmitting critical messages, some blue devices transmitting periodic messages, and the propagations of learning sequences between devices within communication range.}}
	\vspace{-0.3 cm}\label{fig:slpic}
\end{figure}

Consider the uplink of a wireless IoT system consisting of one base station (BS) serving $N$ IoT devices. To transmit their messages to the BS, the IoT devices must first request uplink communication resources via a time-slotted random access channel (RACH) using one out of $p$ random access (RA) preambles \cite{rach} that are typically divided into two types: $p_c$ contention-based RA preambles (RAPs) and $p_f = p - p_c$ contention-free RAPs \cite{rach}. For using a contention-based RAP, the IoT devices choose one out of the $p_c$ preambles randomly. However, if more than one IoT device choose a given contention-based RAP, there will be a collision resulting in a failure to request uplink resources. Therefore, a device will successfully be allocated uplink resources using contention-based RAP, only if it is the sole device using that preamble at that given time slot. For using contention-free RAPs, the devices are assigned one of the $p_f$ preambles by the BS. The preamble usage and allocation among the IoT devices will depend on their message types as discussed later in this section.\\
\indent One of the prominent features of the IoT is heterogeneity in terms of device types, functionalities, and transmitted messages as illustrated in Fig. \ref{fig:slpic}. For instance, some IoT devices, such as smart meters and environment sensors, will be periodically transmitting small packets, such as meter readings, observation reports, and system status reports. However, such devices may need to also transmit urgent, critical messages, such as critical system state condition, power outage, and fire detection reports. These message types differ in terms of quality-of-service (QoS) requirements and traffic patterns. To model heterogeneous messages, we consider the co-existence of \emph{periodic messages}, which are infrequent, periodic transmissions, and \emph{critical messages}, which are critical, urgent transmissions that require very small delay. Depending on the message type, the IoT devices will transmit using most appropriate RAP type.

\begin{table}[t]
\centering\vspace{0mm}
\caption{Summary of notations.}
\label{table}
\begin{tabular}{|l|l|}
\hline
$N$         & Number of IoT devices                                   \\ \hline
$p$         & Number of RAPs                                      \\ \hline
$p_c$, $p_c'$       &\begin{tabular}[c]{@{}l@{}}Number of contention-based RAPs before and after \\ reallocation\end{tabular}   \\ \hline
$p_f$, $p_f'$     & \begin{tabular}[c]{@{}l@{}}Number of contention-free RAPs before and after\\ reallocation\end{tabular}\\ \hline
$T$ & Period for periodic message transmissions \\ \hline
$\tau_i$ & Time slot for first periodic message transmission \\ \hline
$T_{\textrm{min}}$ & Minimum feasible period for periodic messages \\ \hline
$N_p$ & \begin{tabular}[c]{@{}l@{}}Expected number of IoT devices transmitting periodic\\ message at any given time\end{tabular} \\ \hline
$p_{p,s}$ & \begin{tabular}[c]{@{}l@{}}Probability of an IoT device successfully transmitting \\periodic message using contention-based RAP\end{tabular} \\ \hline
$N_a$ & Number of IoT devices transmitting critical message \\ \hline
$p_{a,s}$ & \begin{tabular}[c]{@{}l@{}}Probability of an IoT device successfully transmitting \\ critical message using contention-based RAP\end{tabular} \\ \hline
$D_c$, $D_c'$ & \begin{tabular}[c]{@{}l@{}}Expected delay of critical message under contention-based \\ RAP before and after reallocation \end{tabular}\\ \hline
$D_f$, $D_f'$ & \begin{tabular}[c]{@{}l@{}}Expected delay of critical message under contention-free \\ RAP before and after reallocation \end{tabular} \\ \hline
$D$, $D'$ & Expected delay of a critical message \\ \hline 
$r_d$, $r_d'$ & Detection range for IoT devices before and after learning\\ \hline
$p_{01}$ & Probability of missed detection within detection range \\ \hline
$p_{01}'$ & Probability of missed detection outside of detection range \\ \hline
$\beta$ & Number of contention-free RAPs that are reallocated \\ \hline
$D_{\textrm{th}}$ & Required delay threshold for critical messages \\ \hline
$r_c$ & Communication range for IoT devices \\ \hline
$H_T$ & True underlying state \\ \hline
$s_f$ & Favored state chosen during first phase \\ \hline
$s_f'$ & Any other state that is not favored state \\ \hline
$K$ & Number of observations necessary for first phase \\ \hline
$\alpha$ & \begin{tabular}[c]{@{}l@{}}Number of consecutive observations necessary to go back \\to first phase \end{tabular}\\ \hline
$m$ & Number of bits necessary to learn in second phase\\ \hline
$x_i$ & Private belief of IoT device $i$ based on observations \\ \hline
\end{tabular}\vspace{-4mm}
\end{table}

\subsection{Periodic Messages}
\indent Periodic messages are non-critical and recurring, and, thus, they are typically delay-tolerant with predictable traffic patterns \cite{lat2}. {Moreover, as illustrated by the blue devices in Fig. \ref{fig:slpic}, only some of the IoT devices transmit the periodic messages simultaneously.} Hence, the IoT devices transmitting periodic messages will request uplink resources using an RAP once every period $T$. We let $\tau_i$ for $i = 1, \cdots, T$ be the possible time slots when the IoT devices first transmit their periodic messages. The devices that first transmit in slot $\tau_i$ will transmit in slots $\tau_i + kT$ for $k \in \mathbb{Z}_+$. Since the IoT devices that have the same $\tau_i$ will always transmit simultaneously, it is sufficient to only consider RAP allocation among the IoT devices with same $\tau_i$.
\indent Since it is predictable when and which IoT devices will transmit the periodic messages, contention-free RAPs are more suitable for periodic message transmissions than contention-based RAPs as the contention-free RAPs can guarantee successful uplink resource allocation request for appropriate values of $T$. However, for given values of $N$ and $p_f$, it may be impossible to satisfy $T$, and the minimum period $T_{\textrm{min}}$ of the periodic messages that can be supported will be:
\begin{equation}
T_{\textrm{min}} = \left\lceil\frac{N}{p_f}\right\rceil. \label{tmin}
\end{equation}
\eqref{tmin} results from the fact that at most $p_f$ devices can use contention-free RAPs without collision in a time slot, and an IoT device must wait until all other IoT devices transmit their periodic messages to satisfy $T$ of all periodic messages. If the required period of the periodic messages is less than $T_{\textrm{min}}$, there will be more than $p_f$ IoT devices transmitting simultaneously in some time slots and collisions will occur. For $T \geq T_{\textrm{min}}$, the BS can determine $T$ groups of IoT devices each group with at most $p_f$ IoT devices with same $\tau_i$. Moreover, the BS will allocate different contention-free RAPs to IoT devices with same $\tau_i$. This allocation scheme helps avoiding RAP collision, while maximizing the number of IoT devices that can request uplink resources. The use of contention-free RAPs takes advantage of the predictable traffic patterns of the periodic messages and guarantees that the IoT devices successfully request their uplink resources.\\
\indent Note that contention-based RAPs are not suitable for the periodic messages, because the probability of successful transmission is low, particularly for a massive IoT with large $N$. Assuming that $\tau_i$ are chosen uniformly randomly, if the IoT devices transmitting periodic messages use contention-based RAPs, the expected number $N_p$ of IoT devices transmitting periodic messages at any given time slot will be $N_p = \frac{N}{T}$. The probability $p_{p,s}$ of an IoT device successfully transmitting using the contention-based RAP in a given time slot will then be:
\begin{equation}
p_{p,s} = \left(\frac{p_c - 1}{p_c}\right)^{N_p - 1},
\end{equation}
which represents the probability of all other $N_p - 1$ devices choosing some other contention-based RAP with probability $\frac{p_c - 1}{p_c}$.\\
\indent {Since the IoT will be massive with large $N$, $p_{p,s}$ is small, which implies that very few periodic messages will be successfully transmitted in a given time slot.} For periodic messages with $T \geq T_{\textrm{min}}$, the use of contention-free RAPs guarantees that the IoT devices successfully request uplink resources, while the probability of successful request is low using contention-based RAPs. Therefore, it is more advantageous for the IoT devices transmitting periodic messages to use contention-free RAPs. Hereinafter, we assume that the period $T$ of periodic messages is $T_{\textrm{min}}$ and that $N$ is a multiple of $p_f$ for simplicity. This implies that all $p_f$ contention-free RAPs are used in every time slot.
\subsection{Critical Messages}
\indent Critical IoT device messages contain urgent information about an abnormal event and, thus, they must reach the BS with low latency. Therefore, critical messages are delay-intolerant messages that need retransmissions in case of RAP collision. {Since critical messages are triggered by an unpredictable abnormal event as illustrated in Fig. \ref{fig:slpic}, such as a system failure or forest fire, the system cannot know beforehand which IoT devices will have critical messages to send and when such critical messages will occur.} Furthermore, critical messages can be highly correlated since an abnormal event will typically trigger some critical messages simultaneously. For instance, a forest fire will trigger various devices monitoring different environmental parameters, such as temperature, humidity, or carbon monoxide. Since the network cannot predict when the critical messages will be triggered, how many will be triggered, and which IoT devices will have critical messages, the BS cannot assign pre-determined RAPs, and, thus, IoT devices with critical messages must use contention-based RAPs to request uplink transmission resources.\\
\indent We let $N_a$ be the number of IoT devices having critical messages to send. {$N_a$ depends on the nature and the location of the unpredictable abnormal event. Therefore, $N_a$ is a realization of a random variable, representing an unpredictable abnormal event, and $N_a$ is treated as a constant number.} For IoT devices with critical messages using $p_c$ contention-based RAPs, the probability $p_{a,s}$ of a device being successfully allocated an uplink resource for sending a critical message will be:
\begin{equation}
p_{a,s} = \left(\frac{p_c - 1}{p_c}\right)^{N_a - 1}, \label{critsuccprob}
\end{equation}
which is probability of all other $N_a - 1$ devices choosing some other contention-based RAP with probability $\frac{p_c-1}{p_c}$. {Given $p_{a,s}$, the expected delay $D_c$ of a critical message under a contention-based RAP is defined as the expected number of time slots needed until the first successful acquisition of uplink communication resources using contention-based RAP, as follows:}
\begin{equation}
D_c = p_{a,s}^{-1} = \left(\frac{p_c}{p_c - 1}\right)^{N_a - 1},
\end{equation}
since $D_c$ is a mean of geometric distribution with probability of success of $p_{a,s}$.\\
\indent While using contention-based RAPs, an IoT device having a critical message may be assigned with contention-free RAP for its next periodic message transmission. In this case, instead of using the assigned contention-free RAP for the periodic message, IoT devices with critical messages will use the assigned contention-free RAP for sending their critical messages. {We assume that the period of all periodic messages is smallest possible value of $T_{\textrm{min}}$ and $\tau_i$ are chosen uniformly randomly. Therefore, the expected delay $D_f$ of a critical message under a contention-free RAP is defined as the expected number of time slots needed until an IoT device with critical message is allocated with contention-free RAP, as follows:} \begin{equation}
D_f = \frac{1}{T_{\textrm{min}}} \sum\limits_{j = 0}^{T_{\textrm{min}}-1} j = \frac{T_{\textrm{min}}-1}{2}.
\end{equation}
Since an IoT device having a critical message just needs one successful transmission using either RAP, the expected delay of a critical message is $D = \textrm{min}(D_c, D_f)$.\\
\indent To minimize the expected delay $D$ of the critical messages, $D_c$ or $D_f$ must be minimized. However, decreasing the value of one of them will directly increase the value of the other as $D_c$ and $D_f$ are related with $p_c + p_f = p$. In a massive IoT network with $N \gg p$, $T_{\textrm{min}}$ in \eqref{tmin} will be large even if $p_f = p$, and, thus, $D_f$ will also be large. Since the number of critical messages $N_a$ will typically be small, it is more effective to minimize $D_c$ to minimize $D$. Since the value of $N_a$ cannot be controlled, the value of $p_c$ must be maximized to minimize $D_c$. Since the number of RAPs $p$ is fixed, the only way to increase the value of $p_c$ is to reallocate RAPs from contention-free to contention-based. The method used for such reallocation cannot be centralized since a centralized method requires the BS to have already received the critical messages and to be aware of the abnormal event, which is impractical. Therefore, RAP reallocation must be done in a distributed manner.\\
\indent To reallocate an appropriate number of contention-free RAPs, the IoT devices must be able to observe the number of critical messages $N_a$ in the system. We assume that the overhead signaling for the periodic messages and critical messages are different such that the IoT devices can observe the number of critical messages $N_a$ being transmitted in a given time slot \cite{signal}. However, the IoT devices may not be able to accurately observe $N_a$, since they cannot observe all of the critical messages. For instance, an IoT device that is located far from the abnormal event may not be able to observe the event and its corresponding critical messages. IoT devices may also be prone to missed detection even when they are in proximity to the abnormal event. To capture the limited observation capability, we assume that IoT devices within \emph{a detection range} $r_d$ have a probability of missed detection $p_{01}$, while IoT devices outside of the detection range $r_d$ have a much higher probability of missed detection, $p_{01}'$.\\
\indent In order for the IoT devices to accurately know $N_a$ despite their limited observation capability, they can employ a distributed learning process \cite{fm1} using which the devices can collectively learn the number of critical messages $N_a$ in a self-organizing manner. {For instance, IoT devices with limited observation capabilities can use the observations of other IoT devices to have a more accurate estimation of $N_a$.} For the reallocation of RAPs, we assume that there is \emph{a globally known order of contention-free RAPs that will be reallocated to contention-based RAPs.} When an IoT device having a periodic message learns that there are $N_a$ critical messages, it will stop using the first $\beta$ contention-free RAPs of the order. This means that even if the IoT device with periodic message is assigned to use one of the $\beta$ RAPs by the BS, it will not transmit. When an IoT device having a critical message learns that there are $N_a$ critical messages, it will use one of the first $\beta$ contention-free RAPs of the order or one of the contention-based RAPs. Therefore, the number of contention-based RAPs will increase from $p_c$ to $p_c' = p_c + \beta$, and the number of contention-free RAPs will decrease from $p_f$ to $p_f' = p_f - \beta$. The expected delay $D_c'$ from using contention-based RAP, the expected delay $D_f'$ from using contention-free RAP, and the expected delay $D'$ of critical message \emph{after learning} will, respectively, be given by:
\begin{equation}
D' = \textrm{min}(D_c', D_f') = \textrm{min}\left(\left(\frac{p_c'}{p_c' - 1}\right)^{N_a-1}, \frac{\left\lceil\frac{N}{p_f'}\right\rceil-1}{2}\right). \label{newdelay2}
\end{equation}
For a given design parameter threshold delay $D_{\textrm{th}}$ that the critical messages should satisfy, the value of $\beta$ is determined to satisfy $D' \leq D_{\textrm{th}}$. Moreover, the value of $\beta$ should be appropriate as an excessive value of $\beta$ will greatly disturb the transmission of periodic messages.\\
\indent To develop an effective learning approach, one must take into account specific characteristics of the IoT devices. As the IoT devices are small, low cost devices, their computational capability will be limited. Moreover, they are constrained by very limited battery life, and, thus, they cannot communicate frequently with the BS or with other IoT devices. Further, the IoT devices may also be limited in terms of how much information they can store and process. Therefore, they may not have sufficient information for accurate learning, and their learning method must account for the lack of memory. Also, the IoT devices may have inaccurate information about $N_a$ due to their limited observability of the environment. For applications that require low latency, the expected delay of critical messages must be reduced in relatively real time, and, thus, the IoT devices must quickly and accurately learn $N_a$. Therefore, the learning scheme used by the IoT devices must be computationally simple, distributed, quick, and accurate without requiring excessive memory and frequent communication. To this end, one can use the framework of \emph{finite memory sequential learning} \cite{cover}. Even though this framework has been used before in decentralized, binary decision making, this cannot be readily applied to the IoT as it is only limited to binary state learning. Moreover, the resource constraints of IoT devices in terms of computation and memory, the limited communication resources, and the unique properties of IoT system must be considered. Hence, next we develop a new sequential learning framework tailored to the studied IoT system.
\section{Finite Memory Multi-state Sequential Learning}\label{sec:learning}
For our model, the IoT devices with their inherent limited memory and computation capabilities must learn to satisfy their heterogeneous QoS requirements, in presence of both known periodic messages and unknown critical messages. To this end, we propose a novel multi-state sequential learning framework for enabling the IoT devices to learn the true value of $N_a$ to determine the number of RAPs to reallocate $\beta$ to satisfy the delay threshold $D_{\textrm{th}}$ of critical messages. The IoT devices learning the true value of $N_a$ can be mapped to agents learning the true state of a system with multiple states. For instance, the different values of $N_a$ that the IoT devices can learn can be mapped to different system states, while the true value of $N_a$ is the true, underlying state that the IoT devices should learn. By mapping different values of $N_a$ to different states, we propose \emph{a multi-state learning framework for dynamic reallocation of RAPs}, while taking into account the requirements of the IoT devices. {Furthermore, our learning method only requires an IoT device to be able to communicate with at least one neighboring device within the communication range $r_c$. An IoT device will receive a finite memory of observations of other devices from any one of the neighboring devices. This finite accumulation and flow of information is illustrated in Fig. \ref{fig:slpic} and further discussed in Section III.A and III.B.}\\
\begin{figure}
\centering
\includegraphics[scale=.45]{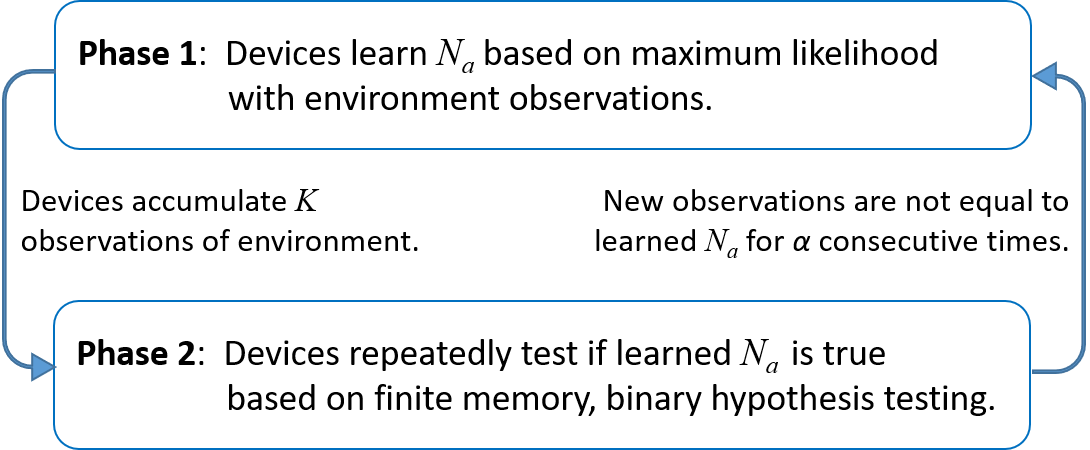}\vspace{-.3cm}
\caption{Flowchart of the proposed finite memory multi-state learning.}\vspace{-.5cm}
\label{flowchart}
\end{figure}
\indent As shown in Fig. \ref{flowchart}, we propose a learning algorithm composed of two phases, which are used to choose a state and to test whether the chosen state is the true underlying state $H_T$. A state is chosen out of all possible states during the first phase, and the chosen state is called a \textit{favored state}. During the second phase, it is repeatedly tested whether the favored state should be changed by going back to the first phase. {In our model, the IoT devices learn $N_a$ that is most likely to be the true value based on their observations and the observations of other devices during the first phase. This learned $N_a$ is a favored state. During the second phase, the following IoT devices will be informed about the learned $N_a$ from the first phase and will repeatedly test whether the learned $N_a$ is a true value or not.} If it is determined that the learned value is not the true value of $N_a$, then the learning method reverts back to the first phase. Moreover, an IoT device belongs to either first phase or second phase and only learns once to prevent the learning sequence from oscillating.\\ 
\indent {The sequential learning scheme is implemented \emph{individually by all devices}, and the IoT devices will learn the true $N_a$ progressively as they get necessary information from the neighboring devices. Moreover, the proposed learning framework does not require a training data set before the devices use the proposed learning framework to learn the true $N_a$. Furthermore, the IoT devices do not converge to the final decision on $N_a$ simultaneously. Although the devices may learn $N_a$ incorrectly due to limited observability, the effect of limited observability can be mitigated by increasing the memory size $m$ to extend the detection radius as discussed in Section III.B. Furthermore, the devices cannot determine the validity of learned value of $N_a$, but the proposed learning framework is proved to converge to the true value of $N_a$ in probability as discussed in Section III.C. As learning propagates sequentially, the IoT devices adjust their RAP usage as soon as they learn $N_a$, and, thus, the reallocation of $\beta$ contention-free RAPs is done gradually.} Moreover, $\beta$ depends on the learned value of $N_a$ and the given value of design parameter $D_{\textrm{th}}$ as further discussed in Subsection III.D.\\
\indent For both phases, the observations are critical in learning $N_a$ correctly, and more observations typically result in higher probability of learning correctly. However, it is unrealistic to assume that an IoT device is able to know, store, and process observations of all other IoT devices, which corresponds to \emph{infinite memory} sequential learning in \cite{cover}. Therefore, our proposed learning framework has \emph{finite memory} \cite{fm1} as it assumes that an IoT device only utilizes some of the observations of other IoT devices. In other words, the finite memory learning requires finite number of observations, while the infinite memory learning requires infinite number of observations. The convergence of our finite memory learning framework is shown in Subsection III.C.
\subsection{First Phase of Sequential Learning}
For an appropriate reallocation of contention-free RAPs, the IoT devices must learn the true value of the number of critical messages $N_a$. However, the devices may learn wrong values of $N_a$ due to limited observability. {We let $\mathcal{S}$ be the set of all possible values of $N_a$, and $\mathcal{S}$ is a set of all possible states.} Here, the true value of $N_a$ is the true underlying state $H_T \in \mathcal{S}$. {During the first phase, the IoT devices chose a favored state $s_f \in \mathcal{S}$ after $K$ observations, which provide information about $H_T$.} {$s_f$ is a value that the devices learn to be $H_T$. However, the devices do not know if $s_f = H_T$ or $s_f \neq H_T$, which is repeatedly tested during the second phase of learning. Moreover, $K \ll N$ as we assume that the IoT devices are limited in terms of memory and computational capabilities.}\\
\indent In our model, the IoT devices observe how many critical messages exist depending on their limited observation range $r_d$. For simplicity, we model an observation of the environment $e$ to be an element of $\mathcal{S}$ such that, for any given $H_T \in \mathcal{S}$, 
\begin{align}
0 < \frac{\Pr(e = j \mid H_T)}{\Pr(e = k \mid H_T)} &< \infty \ \forall j, k \in \mathcal{S}\label{likelihood1},\\
\Pr(e = H_T \mid H_T) &\gg \Pr(e \neq H_T \mid H_T)\nonumber,\\
\Pr(e = i \mid i) &\gg \Pr(e = i \mid j) \ \forall i, j \in \mathcal{S} \label{assump1},\\
\Pr(e = j \mid i) &\approx \Pr(e = k \mid i) \ \forall i, j, k \in \mathcal{S} \label{assump3},
\end{align}
where $\Pr(e = i | j)$ is probability of observing $i$ when $j$ is the true underlying state. \eqref{likelihood1} is necessary to ensure that the learning process is not trivial by preventing the likelihood ratio of observations from being $0$ or $\infty$. In \eqref{assump1}, we assume that a given device is more likely to observe the true underlying state $H_T$ than any other state in $\mathcal{S}$ and that it is most likely to observe $i \in \mathcal{S}$ when the true underlying state is $i$. Furthermore, in \eqref{assump3}, the probabilities of observing the states that are not true underlying state are approximately equal. We assume that the environment observations, which correspond to the observed numbers of critical messages, are independent, because the observability of an IoT device depends only on the observation range $r_d$.\\
\indent The first phase is initiated by an abnormal event, and the devices that are first to observe the abnormal event will have to transmit critical messages. {Furthermore, as illustrated by the red devices in Fig. \ref{fig:slpic}, those devices propagate the learning sequence to their neighboring devices within $r_c$ by transmitting their observations of the environment.} Those neighboring devices constitute the next set of devices that learn in the first phase, and these devices learn $N_a$ based on their observations and the information from the device that learned first. {Therefore, the first phase propagates sequentially to the various IoT devices within $r_c$, forming a learning sequence as shown in Fig. \ref{fig:slpic}.} When a device obtains $K$ observations, it determines $s_f$ based on $K$ observations and initiates the second phase. {If an IoT device that has not yet learned receives the necessary information for learning from more than one neighboring device simultaneously, then it will randomly choose one of the learning sequences to continue.}\\
\indent We assume that $r_c < r_d$ as the communication links between IoT devices are usually short-ranged. Moreover, it is desirable for $K$ to be small so that $s_f$ is chosen quickly and that the learning method converges faster to $H_T$. For small $K$, the memory and computational requirements for the IoT devices will also be small due to the reduced amount of information that must be stored and processed, and, thus, the critical messages will be successfully transmitted with small delay as the learning converges faster to $H_T$. Hence, the IoT devices during the first phase will most likely be within $r_d$. Since the probability of missed detection $p_{01}$ within $r_d$ is usually close to $0$ \cite{learning6}, the IoT devices participating in the first phase are most likely to observe $e = H_T$. Therefore, \eqref{assump1} holds true for the IoT devices within $r_d$, where the first phase always occurs.\\
\indent {While the IoT devices accumulate $K$ independent observations of the environment, they can make an intelligent decision on which state is most likely to be the true underlying state $H_T$, which is true value of $N_a$.} This decision will be based on \textit{maximum likelihood}, whose computational complexity can be simplified given \eqref{assump1}. For a set of $K$ observations $\{e_1, \cdots, e_K\}$, the favored state $s_f$ will be:
{\vspace{-2mm}\begin{align}
s_f &= \argmax\limits_{s \in \mathcal{S}} \Pr(e_1, \cdots, e_K \mid s),\\
&= \argmax\limits_{s \in \mathcal{S}} \prod\limits_{j=1}^{K} {\Pr(e_j \mid s)} = \argmax\limits_{s \in \mathcal{S}} \prod\limits_{j\in\mathcal{S}} \Pr(j \mid s)^{k_j},\\
&= \argmax\limits_{s \in \mathcal{S}} \sum\limits_{j\in\mathcal{S}} {k_j}\ln(\Pr(j\mid s)), \label{ml}
\end{align}}
where $k_j$ is the number of times that state $j \in \mathcal{S}$ is observed out of $K$ observations. If the likelihoods $\Pr(e = i\mid j) \ \forall i, j \in \mathcal{S}$ are known by the devices, then $s_f$ can be determined easily using \eqref{ml}. However, when the likelihoods $\Pr(e = i\mid j) \ \forall i, j \in \mathcal{S}$ are unknown, the maximum likelihood can be simplified using \eqref{assump1} and \eqref{assump3} and, then, the favored state $s_f$ will be $a \in \mathcal{S}$ if:
\begin{equation}
\frac{\sum\limits_{j\in\mathcal{S}} {k_j} \ln(\Pr(j \mid a))}{\sum\limits_{j\in\mathcal{S}} {k_j} \ln(\Pr(j \mid b))} \approx \frac{\sum\limits_{j\in \mathcal{S}, j\neq a} k_j}{\sum\limits_{j\in \mathcal{S}, j \neq b} k_j} = \frac{K - k_a}{K - k_b} \leq 1 \ \forall \ b \in \mathcal{S}.
\end{equation}
Therefore, with simplifications, the decision process in the first phase will be such that the devices will be able to designate one of the observed states as a favored state $s_f$.\\
\indent {The communication overhead involved in the first phase of sequential learning pertains to each IoT device transmitting at most $K-1$ observations to its neighboring devices only once. Moreover, based on the aforementioned simplifications using \eqref{assump1} and \eqref{assump3}, the computational complexity involved in the first phase of sequential learning is $O(K)$, which is reasonable.}
\subsection{Second Phase of Sequential Learning}
{Once the IoT devices learn $s_f$ to be the true value of $N_a$ after $K$ observations, the validity of $s_f$ must be repeatedly tested, given the possibility of learning error due to limited observability.} If it is found that $s_f \neq H_T$, then $s_f$ must be changed to some other state in $\mathcal{S}' = \mathcal{S}\setminus \{s_f\}$. Since the second phase tests whether $s_f$ must be changed or not, this can be modeled as binary hypothesis testing between $s_f = H_T$ and $s = H_T$ for some state $s \neq s_f$, $s \in \mathcal{S}'$, and the IoT devices in second phase each test with their available information. This binary hypothesis testing uses binary environmental observations to test whether the favored hypothesis should be changed, and the favored hypothesis will converge to the underlying truth by design.\\
\indent As our learning method is multi-state, the devices will observe more than two states. However, if a device observes any state in $\mathcal{S}'$, it will be equivalent to observing an unfavored state $s_f'$. Therefore, the number of possible observations reduces to two: $s_f$ and $s_f'$, during the second phase. Here, observing $s_f$ is equivalent to $1$ and observing $s_f'$ is equivalent to $0$. Similar to the observations in first phase, an environment observation $e$ in the second phase is such that, for an underlying true state $H_T \in \mathcal{S}$, 
\begin{align}
\Pr(e = s_f& \mid H_T) + \Pr(e \in \mathcal{S}' \mid H_T) = 1 \\ 0 &< \frac{\Pr(e = s_f \mid H_T)}{\Pr(e = s_f' \mid H_T)} < \infty. \label{likelihood2}
\end{align}
\indent If $s_f \neq H_T$ is chosen during the first phase and the environment observations repeatedly indicate $s_f'$, then the favored state must be changed. {Therefore, a mechanism for returning to the first phase to change $s_f$ is necessary.} We set a threshold $\alpha$ such that if the environment observations indicate $s_f'$ for a total of $\alpha$ consecutive times, then the learning method will go back to the first phase to choose $s_f$ again with some probability $p_b$, which will be further discussed later. Using a threshold, our learning method may go back to the first phase both when $s_f = H_T$ and when $s_f \neq H_T$, and the probabilities of returning to the first phase for both cases are critical for the convergence of the proposed learning method.\\
\indent The probability of observing $s_f'$ for $\alpha$ consecutive times when $s_f \neq H_T$ is:
{\vspace{-2mm}{\begin{align}
\Pr(\{e_{n_1-\alpha+1}, \cdots, e_{n_1}\} &= \{s_f', \cdots, s_f'\}  \forall n_1 \geq \alpha \mid s_f \neq H_T)\\
&= (1 - \Pr(e = s_f \mid s_f \neq H_T))^{\alpha} \label{prob1},
\end{align}}}
\noindent where $e_j$ is $j$-th observation of the environment and $n_1 \in \mathbb{Z}_+$ is the number of observations necessary for $\alpha$ consecutive observations of $s_f'$ to occur when $s_f \neq H_T$. $n_1$ is a random variable, and it is essential that the learning method quickly reverts to the first phase if $s_f \neq H_T$. {Therefore, the expected value $\mathbb{E}[n_1\mid s_f \neq H_T]$ should be small.} If $s_f$ is observed after $s_f'$ is observed $j < \alpha$ consecutive times, then the number of consecutive observations of $s_f'$ is reset to $0$, and this is equivalent to increasing $\mathbb{E}[n_1 \mid s_f \neq H_T]$ by $j+1$. With $p_1 = \Pr(e = s_f \mid s_f \neq H_T)$, the aforementioned case occurs with probability $p_1(1-p_1)^j$. Therefore, the expected value of $n_1$, for any environment observation $e$, will be:
{
\vspace{0mm}{{\begin{align}
&\mathbb{E}[n_1 \mid s_f \neq H_T]\\
&= \sum\limits_{j = 0}^{\alpha-1} p_1(1 - p_1)^j \left(\mathbb{E}[n_1 \mid s_f \neq H_T]+j+1\right) + \alpha(1 - p_1)^\alpha,\\
&= \mathbb{E}[n_1 \mid s_f \neq H_T]p_1\left(\sum\limits_{j=0}^{\alpha-1} (1 - p_1)^j \right) \nonumber\\
& \hspace{1.2cm}+p_1\left(\sum\limits_{j=0}^{\alpha-1} (1 - p_1)^j(j+1) \right) + \alpha(1-p_1)^\alpha,\\
&= \mathbb{E}[n_1 \mid s_f \neq H_T]p_1\frac{1 - (1 - p_1)^\alpha}{1 - (1 - p_1)} \nonumber \\
& \hspace{1.2cm}+p_1\left(\sum\limits_{j=0}^{\alpha-1} (1 - p_1)^j(j+1) \right) + \alpha(1-p_1)^\alpha,\\
&= \frac{p_1\left(\sum\limits_{j=0}^{\alpha-1} (1 - p_1)^j(j+1) \right) + \alpha(1-p_1)^\alpha}{(1 - p_1)^\alpha}. \label{exp1}
\end{align}}}}
\indent Similarly, the probability of observing $s_f'$ for $\alpha$ consecutive times when $s_f = H_T$ is
{\begin{align}
\Pr(\{e_{n_2-\alpha+1}, \cdots, e_{n_2}\} &= \{s_f', \cdots, s_f'\} \forall n_2 \geq \alpha \mid s_f = H_T)\\
&= (1 - \Pr(e = s_f \mid s_f = H_T))^{\alpha} \label{prob2},
\end{align}}
where $n_2 \in \mathbb{Z}_+$ is the number of observations necessary for $\alpha$ consecutive observations of $s_f'$ to occur when $s_f = H_T$. Similar to $n_1$, $n_2$ is a random variable, but it must be highly unlikely that the learning method goes back to the first phase if $s_f = H_T$. In other words, the expected value $\mathbb{E}[n_2\mid s_f = H_T]$ should be very large. If $s_f$ is observed after $s_f'$ is observed $j < \alpha$ consecutive times, then the number of consecutive observations of $s_f'$ is reset to $0$, and this is equivalent to increasing $\mathbb{E}[n_2 \mid s_f = H_T]$ by $j+1$. With $p_2 = \Pr(e = s_f \mid s_f = H_T)$, the aforementioned case occurs with probability $p_2(1-p_2)^j$. Therefore, the expected value of $n_2$, for any environment observation $e$, will be: 
{
\vspace{0mm}{{\begin{align}
&\mathbb{E}[n_2 \mid s_f = H_T] \\
&= \sum\limits_{j = 0}^{\alpha-1} p_2(1 - p_2)^j \left(\mathbb{E}[n_2 \mid s_f = H_T]+j+1\right) + \alpha(1 - p_2)^\alpha,\\
&= \mathbb{E}[n_2 \mid s_f = H_T]p_2\left(\sum\limits_{j=0}^{\alpha-1} (1 - p_2)^j \right) \nonumber\\
&\hspace{1.2cm} + p_2\left(\sum\limits_{j=0}^{\alpha-1} (1 - p_2)^j(j+1) \right) + \alpha(1-p_2)^\alpha,\\
&= \mathbb{E}[n_2 \mid s_f = H_T]p_2\frac{1 - (1 - p_2)^\alpha}{1 - (1 - p_2)} \nonumber \\
&\hspace{1.2cm} + p_2\left(\sum\limits_{j=0}^{\alpha-1} (1 - p_2)^j(j+1) \right) + \alpha(1-p_2)^\alpha,\\
&= \frac{p_2\left(\sum\limits_{j=0}^{\alpha-1} (1 - p_2)^j(j+1) \right) + \alpha(1-p_2)^\alpha}{(1 - p_2)^\alpha} \label{exp2}.
\end{align}}}}
\indent The threshold $\alpha$ can be seen as a design parameter to adjust the expected number of necessary iterations in \eqref{exp1} and \eqref{exp2}. Moreover, there are advantages and disadvantages of choosing a high value and a low value for $\alpha$. {Given \eqref{assump1}, a high value of $\alpha$ will imply that it is highly unlikely to incorrectly go back to the first phase when $s_f = H_T$ is chosen.} {However, if $s_f \neq H_T$ is chosen, it will take many observations to go back to the first phase, and, thus, the learning method converges slowly to $s_f = H_T$. A small $\alpha$ implies that the learning method can quickly go back to the first phase to choose $s_f$ again if $s_f \neq H_T$. Moreover, a small $\alpha$ also implies that it is more likely to incorrectly go back to the first phase even if $s_f = H_T$.}\\
\indent Similar to the first phase, the IoT devices in the second phase will individually learn after receiving $m$ bits of information from any neighboring IoT device that learned and, then, they will transmit $m$ bits of updated information to neighboring IoT devices. The \emph{memory} that IoT device $i$ needs to learn $N_a$ is essentially the size of the received information set $\{e_{i-m+2}, \cdots, e_{i-1}, F_{i-1}, Q_{i-1}\}$, where $e_{i-1}$ is the environment observation of IoT device $i-1$, $F_{i-1}$ is the currently favored hypothesis, and $Q_{i-1}$ indicates whether $F_{i-1}$ should be changed. The value of $m$ must be at least two bits to capture $\{F_{i-1}, Q_{i-1}\}$, which are necessary for learning \cite{cover}, and the second phase reduces to the finite memory sequential learning introduced in \cite{cover} if $m = 2$. However, the finite memory sequential learning in \cite{cover} assumes that \eqref{assump1} holds true for all environment observations, which is not true for the IoT devices with limited observability. Therefore, the memory of previous observations is used to mitigate the limited observability.\\ 
\indent For $m > 2$, an IoT device $i$ computes its \emph{private belief} $x_i$ based on maximum likelihood using the observations of previous devices in learning sequence $\{e_{i-m+2}, \cdots, e_{i-1}\}$ and its own observation $e_i$. The private belief $x_i$ of IoT device $i$ is:
\begin{align}
x_i &= \argmax\limits_{s\in\{s_f, s_f'\}} \Pr(e_{i-m+2}, \cdots, e_i \mid s)\\
&= \argmax\limits_{s\in\{s_f, s_f'\}} \prod\limits_{\scriptscriptstyle{j=i-m+2}}^{i} {\textstyle{\Pr(e_j \mid s)}},\\
&= \argmax\limits_{s\in\{s_f, s_f'\}} \Pr(s_f\mid s)^{k_{s_f}}\Pr(s_f' \mid s)^{k_{s_f'}}\\
&= \argmax\limits_{s\in\{s_f, s_f'\}} k_{s_f}\ln(\Pr(s_f\mid s)) + k_{s_f'}\ln(\Pr(s_f' \mid s)),\label{psml}
\end{align}
where $k_{s_f}$ is the number of times that state $s_f$ is observed and $k_{s_f'}$ is the number of times that state $s_f'$ is observed out of $m-1$ observations. If the likelihoods $\Pr(e = i|j) \ \forall i,j \in \{s_f, s_f'\}$ are known by the devices, then $x_i$ can be determined easily using \eqref{psml}. However, when the likelihoods are unknown, the maximum likelihood can be simplified using \eqref{assump1} and \eqref{assump3}, and the private signal is:
\begin{align}
    x_i&=\left\{
                \begin{array}{ll}
                  s_f \hfill &\text{if} \ \frac{k_{s_f}\ln(\Pr(s_f\mid s_f)) + k_{s_f'}\ln(\Pr(s_f' \mid s_f))}{k_{s_f}\ln(\Pr(s_f\mid s_f')) + k_{s_f'}\ln(\Pr(s_f' \mid s_f'))}\leq 1,\\
                  s_f' \hfill &\text{otherwise},
                \end{array}\right.\\
              &= \left\{
                \begin{array}{ll}
                  s_f \hfill &\text{if} \ k_{s_f} \geq 1,\\
                  s_f' \hfill &\text{otherwise}.
                \end{array}\right. \label{pslrt}
\end{align}
Therefore, the private belief is $s_f$ if any of the environment observations is $s_f$.\\
\indent When $s_f \neq H_T$, the $k_{s_f}$ will most likely be $0$, because the observations from both within $r_d$ and outside of $r_d$ will highly unlikely be $s_f$. However, when $s_f = H_T$, $k_{s_f}$ will most likely be greater than $0$ if any of the observations is from within $r_d$. This is because the observations outside of $r_d$ is subject to high probability of missed detection, and, thus, the \eqref{assump1} does not hold. As the oldest observation $e_{i-m+2}$ is replaced by the newest observation $e_i$ and there are $m-2$ previous observations, all of the observations are from outside of $r_d$ after the learning sequence has progressed $m-2$ times outside of $r_d$. Therefore, assuming that the network is dense and most IoT devices have at least one neighboring IoT device, the furthest IoT device from the abnormal event with an observation from within $r_d$ is located $r_d + (m-2)r_c$. Therefore, by using $x_i$, $k_{s_f}$ of the IoT devices within $r_d + (m-2)r_c$ will most likely be greater than $0$ and have $x_i = s_f$ when $s_f = H_T$. In other words, \eqref{assump1} holds for more IoT devices, and more IoT devices will learn correctly. Moreover, the use of $x_i$ instead of $e_i$ is very effective in the second phase, because the learning progresses sequentially away from the abnormal event and the IoT devices participating in the second phase will most likely be outside of $r_d$. The effective detection radius $r_d'$ within which \eqref{assump1} holds can therefore be defined as follows:
\begin{equation}
r_d' = r_d + (m - 2)r_c, \label{effrd}
\end{equation}
where $m$ is the memory size of a device (in bits). For a larger memory size $m$, $r_d'$ is bigger, and, thus, \eqref{assump1} holds for more IoT devices. However, as the deployment region is finite, it is unnecessary to have very large values of $m$, which can render $r_d'$ much greater than the dimensions of the deployment region. {Moreover, for larger $m$, more energy is required to transmit and compute, which may be impractical for IoT devices with energy constraints. Therefore, it is necessary to choose an appropriate value of $m$ depending on deployment region, energy constraints, and delay requirements.}\\
\indent An IoT device $i$ in the second phase learns using $x_i$, $F_{i-1}$, and $Q_{i-1}$ by repeatedly testing whether the currently favored hypothesis $F_{i-1} \in \{0, 1\}$ should be changed. Moreover, a number of consecutive devices in the learning sequence need to collectively test and decide whether $F_{i-1}$ should be changed. Here, we define the notion of an \emph{S-block} as a set consecutive devices testing whether $F_{i-1}$ should be changed from $0$ to $1$ and an \emph{R-block} as a set of consecutive devices testing whether $F_{i-1}$ should be changed from $1$ to $0$. Every device in a learning sequence belongs to either an S-block or an R-block, and both blocks are alternating in a learning sequence to repeatedly test $F_{i-1}$. As shown in \cite{cover}, the lengths of the S and R-blocks must be chosen to ensure that $F_{i-1}$ will converge to true hypothesis.\\
\indent Since observing $s_f$ is equivalent to $1$ and observing $s_f'$ is equivalent to $0$, the IoT devices belonging to an S-block decide that $F_{i-1}$ should be changed from $s_f'$ to $s_f$ if all devices in that S-block observe $s_f$. Similarly, the IoT devices belonging to an R-block decide that $F_{i-1}$ should be changed from $s_f$ to $s_f'$ if all devices in that R-block observe $s_f'$. $Q_{i-1}$ is used to track if all of the observations are $s_f$ in S-block or $s_f'$ in R-block, and $F_{i-1}$ is changed accordingly by the IoT devices. After learning, IoT device $i$ updates $m$ bits of information by replacing $\{F_{i-1}, Q_{i-1}\}$ with $\{F_i, Q_i\}$ and the oldest observation $e_{i-m+2}$ with its own observation $e_{i}$. Moreover, IoT device $i$ propagates the updated information to its neighboring devices for them to learn.\\
\indent {The communication overhead involved in the second phase of the sequential learning pertains to each IoT device transmitting at most $m$ bits of necessary information to its neighboring devices only once, as an IoT device belonging to either phases of learning will not learn again. Since the binary hypothesis testing involving $F_{i-1}$ and $Q_{i-1}$ can be considered as finite state machine \cite{cover} and computing $x_i$ can be simplified using \eqref{assump1} and \eqref{assump3}, the computational complexity involved in the second phase of learning is $O(m)$.}
\subsection{Convergence of the Proposed Learning Method}
Our proposed learning method consists of two phases during which a favored state $s_f$ is chosen and then is repeatedly tested whether $s_f$ should be changed. For the first learning phase, it must be shown that the favored state $s_f = H_T$ can be chosen correctly during the first phase with nonzero probability. Moreover, for the second phase, it must be shown that the learning method will not revert back to first phase with probability approaching 1 if $s_f = H_T$. As $s_f \neq H_T$ can be chosen during the first phase, it must also be shown that the learning method will revert back to first phase with probability 1 when $s_f \neq H_T$. If the aforementioned properties of both first and second phases of the learning method can be shown, then the learning method will be guaranteed to converge to the correct underlying state $H_T$, as shown next. 
\begin{theorem}\label{pro1}
For a memory size $m$ such that $(m-2) \geq \alpha$, if \eqref{assump1} holds, then the proposed finite memory learning method will converge in probability to the true, underlying state $H_T$.
\end{theorem}
\begin{proof} See Appendix \ref{app1}. \end{proof}
\indent The result of Theorem \ref{pro1} applies whenever \eqref{assump1} holds true for all agents participating in the learning method and none of the agents suffers from the limited observation capability. However, the IoT devices in our model have limited observation capabilities, and, thus, \eqref{assump1} may not hold true for all of them. However, with sufficiently large memory size $m$ such that $r_d'$ in \eqref{effrd} is large enough to circumscribe the entire deployment region, \eqref{assump1} will hold true for all IoT devices. Therefore, with an additional condition on the value of $m$, the proposed finite memory learning method will converge in probability to $H_T$ for the IoT devices. Furthermore, the IoT devices are limited in memory, and, thus, the effect of having small or big memory sizes on the performance and convergence of learning will be further analyzed via simulations in Section IV. 
\subsection{Delay of Critical Messages}
For larger sizes $m$ of the devices' memory, \eqref{assump1} holds true for more IoT devices that can learn $\beta$ correctly during the second phase, and, thus, the delay of the IoT devices with critical messages with be smaller. As the learning propagates during $t$ time slots, an IoT device, which learned and is located furthest from a device with critical messages, is located at a distance of $tr_c$ from the abnormal event. In other words, after $t$ time slots, only the IoT devices that are located closer than $tr_c$ are part of a learning sequence. Therefore, whether an IoT device was able to learn correctly depends on both $m$ and $t$ as the IoT devices that are located within $r_t = \textrm{min}(tr_c, r_d')$ are likely to learn correctly.\\
\indent Assuming a massive IoT network with limited number of RAPs with $N \gg p$, $T_{\textrm{min}}$ is very large, and, thus, $D$ and $D'$ are dictated by $D_c$ and $D_c'$ respectively. Since $D_c'$ approaches 1 as $\beta$ increases, $D_{\textrm{th}}$ must be greater than 1 for a feasible value of $\beta$. For a given threshold delay $D_{\textrm{th}} < D_c$, the value of $\beta$ is, given system parameters $p_c$ and $N_a$:
\begin{align}
D_c' = \left(\frac{p_c + \beta}{p_c + \beta -1}\right)^{N_a-1} &\leq D_{\textrm{th}},\\
p_c + \beta &\leq \sqrt[\leftroot{-3}\uproot{3}N_a-1]{D_{\textrm{th}}}(p_c + \beta -1) \\ 
\beta(1 - \sqrt[\leftroot{-3}\uproot{3}N_a-1]{D_{\textrm{th}}}) &\leq \sqrt[\leftroot{-3}\uproot{3}N_a-1]{D_{\textrm{th}}}(p_c - 1) - p_c,\\
\frac{\sqrt[\leftroot{-3}\uproot{3}N_a-1]{D_{\textrm{th}}}(p_c - 1) - p_c}{1 - \sqrt[\leftroot{-3}\uproot{3}N_a-1]{D_{\textrm{th}}}} \leq \beta &= \left\lceil\frac{\sqrt[\leftroot{-3}\uproot{3}N_a-1]{D_{\textrm{th}}}(p_c - 1) - p_c}{1 - \sqrt[\leftroot{-3}\uproot{3}N_a-1]{D_{\textrm{th}}}}\right\rceil \label{beta}.
\end{align}
For the value of $\beta$ in \eqref{beta}, the delay threshold $D_{\textrm{th}}$ will be satisfied for the critical messages for $m$ and $t$ approaching $\infty$, while minimally affecting the periodic transmissions. However, it is impractical to have very large $m$ considering the restricted resources of IoT devices. Furthermore, it is of interest to analyze the transient phase of our learning method, because the IoT devices with critical messages try to transmit as other devices are still learning. Therefore, the maximum number of reallocated RAPs $\beta_t$ at time $t$, for finite values of $m$ and $t$, is of interest to determine the expected delay of critical messages for realistic values of $m$ and $t$. 
\begin{theorem}\label{pro2}
For finite values of $t$ and $m$, the maximum expected number of reallocated RAPs $\beta_t$ at time $t$ with average $n_t$ IoT devices that learned correctly is:
\begin{equation}
\mathbb{E}[\beta_t] = \frac{(p_f - \beta)!}{p_f!} \sum\limits_{b = 0}^\beta b P(n_t, b) P(p_f - n_t, \beta - b) C(\beta, b),
\end{equation}
where $P(n, k)$ is $k$-permutations of $n$ and $C(n,k)$ is $k$-combinations of $n$.
\end{theorem}
\begin{proof} See Appendix \ref{app2}. \end{proof}
\indent Theorem \ref{pro2} derives the maximum expected number of contention-free RAPs that will be reallocated to contention-based RAPs to reduce the delay of critical messages in a realistic IoT with finite $t$ and $m$. Moreover, it is interesting to note that when $m = \infty$ and $t = \infty$, $n_t = p_f$ as $r_d'$ will be large enough to include the entire deployment region. Therefore, the probability of having $\beta_t$ reallocated RAPs at time $t$ in \eqref{betatprob} is $1$ for $\beta_t = \beta$ and zero otherwise, and, thus, $\mathbb{E}[\beta_t] = \beta$. Moreover, for higher values of $m$, $t$, and, thus, $n_t$, $\mathbb{E}[\beta_t]$ will approach $\beta$ and the critical messages are more likely to satisfy $D_{\textrm{th}}$. With finite values of $m$ and $t$, the lowest expected delay that can be achieved with our proposed learning method is $\left(\frac{p_c + \mathbb{E}[\beta_t]}{p_c + \mathbb{E}[\beta_t] - 1}\right)^{N_a}$.
\section{Simulation Results and Analysis}
For our simulations, we consider a rectangular area of width $w$ and length $l$ within which the IoT devices are deployed following a Poisson point process of density $\lambda$ and an expected number of IoT devices $wl\lambda$. We let $w = l = 100$ m, $r_c = 2$ m, and $r_d = 10$ m, and we choose a time slot duration of $0.25$ ms \cite{tslot}. A random location within the deployment region is chosen to be the location of an abnormal event, and the IoT devices within a distance of $1$ m will have critical messages and initiate the learning procedure. The number of RAPs $p$ is $64$ \cite{rach}, which includes $p_f = 63$ contention-free RAPs and $p_c = 1$ contention-based RAPs. We also set $K = 3$ and $\alpha = 5$. {Moreover, the probability of missed detection $p_{01}'$ outside of $r_d$ is set to 0.9 \cite{learning6}}. All statistical results are averaged over a large number of independent runs. 

\begin{figure} 
    \centering
  \subfloat[Learning with smaller memory of $m =3$ bits.]{%
        \includegraphics[width = 9cm]{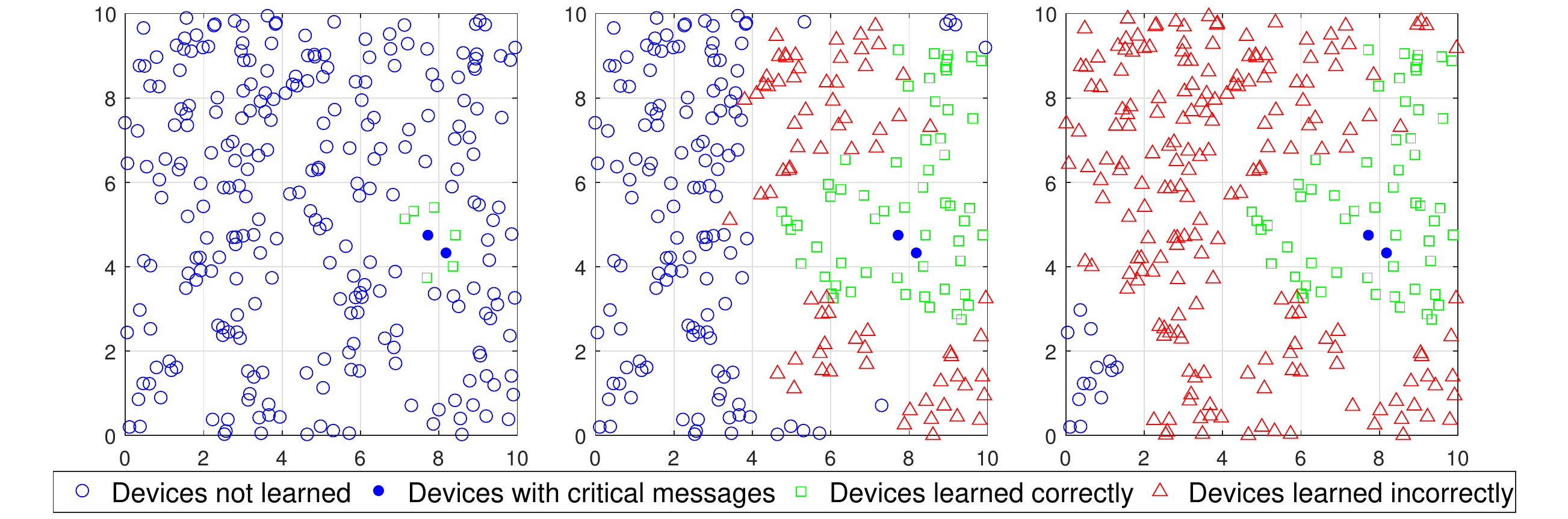}}
    \label{snapshot1}\\
  \subfloat[Learning with larger memory of $m = 10$ bits.]{%
        \includegraphics[width = 9cm]{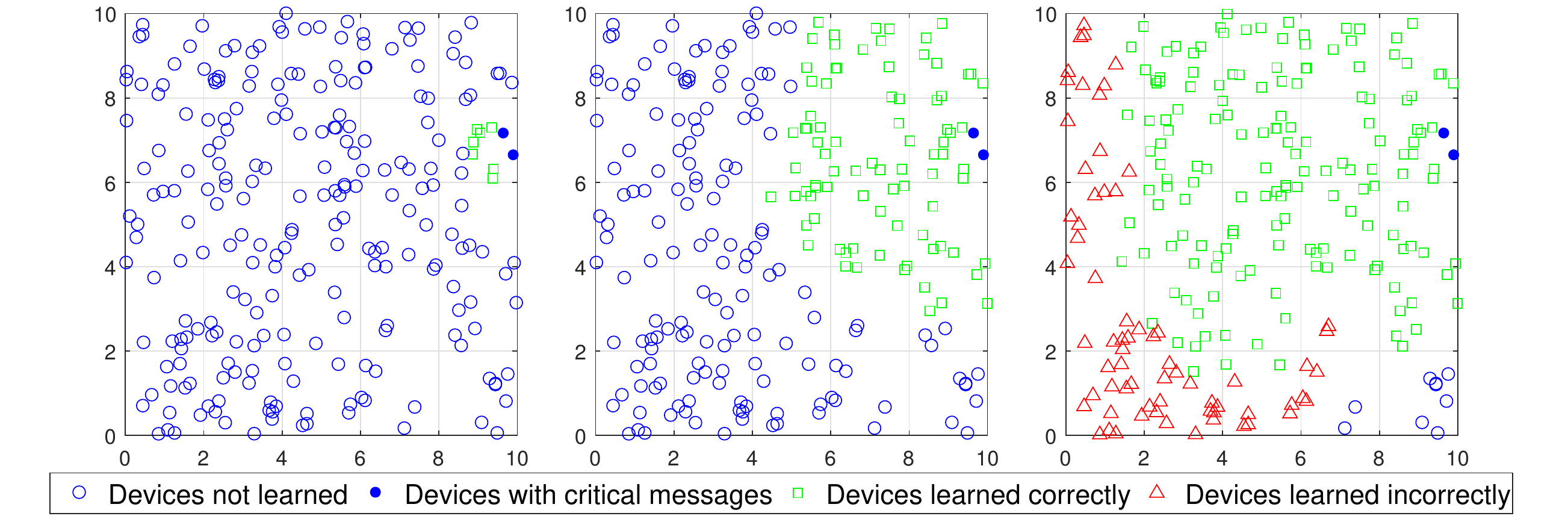}}
    \label{snapshot2}\hfill
  \caption{Snapshots showing how the proposed learning scheme progresses in small scale networks at $t = 0.25$ ms, $1.75$ ms, and $4.25$ ms, respectively.}\vspace{-.5cm}
  \label{snapshots} 
\end{figure}

\indent {Fig. \ref{snapshots} shows snapshots of learning with different memory sizes progressing in a small-scale network with different values of $w$, $l$, $r_c$, and $r_d$ at $t = 0.25$ ms, $1.75$ ms, and $4.25$ ms.} {These times are chosen to show how the proposed learning framework progresses over time from the initiation to the convergence.} For both scenarios, learning converges by $t = 4.25$ ms. At $t = 0.25$ ms, for both values of $m$, only the neighboring devices of the IoT devices with critical messages had a chance to learn in the first phase, and they all learned correctly. At $t = 1.75$ ms, for learning with $m = 3$ bits shown in Fig. \ref{snapshots}(a), IoT devices outside of $r_d'$ started to learn incorrectly due to the limited observability. However, for learning with $m = 10$ bits shown in Fig. \ref{snapshots}(b), IoT devices are still within $r_d'$ and are learning correctly. At $t = 4.25$ ms, even with larger memory size of $m = 10$ bits, some IoT devices are outside of $r_d'$ and are unable to learn correctly. For both memory sizes, there are few IoT devices that did not have a chance to learn as they are not within $r_c$ of any other IoT devices.

\begin{figure}
\centering
\includegraphics[width = 9cm]{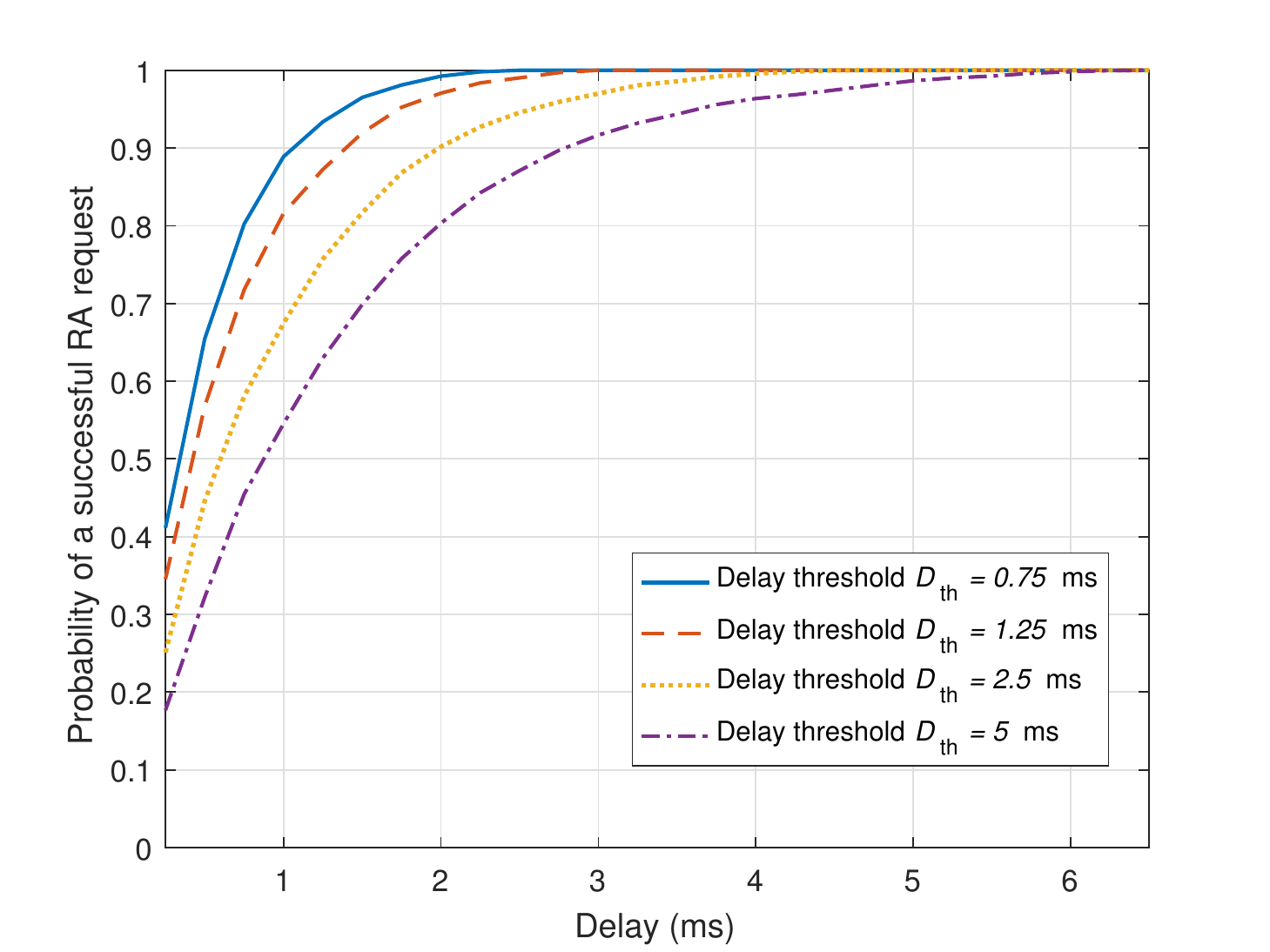}\vspace{-.2cm}
\caption{Cumulative distribution function of delay of critical messages for different $D_{\textrm{th}}$.}\vspace{-.5cm}
\label{latendth}
\end{figure}

\indent Fig. \ref{latendth} shows the cumulative distribution function for the delay of critical messages for different delay thresholds $D_{\textrm{th}}$ with $\lambda = 2$, $p_{11} = 0.9$, and $m = 5$ bits. Fig. \ref{latendth} shows that the delay threshold is satisfied, which means that the achieved delay is less than or equal to $D_{\textrm{th}}$, with probability $0.80$ for $D_{\textrm{th}} = 0.75$ ms, $0.87$ for $D_{\textrm{th}} = 1.25$ ms, $0.95$ for $D_{\textrm{th}} = 2.5$ ms, and $0.99$ for $D_{\textrm{th}} = 5$ ms. Moreover, the average achieved delay is $0.59$ ms for $D_{\textrm{th}} = 0.75$ ms, $0.72$ ms for $D_{\textrm{th}} = 1.25$ ms, $0.99$ ms for $D_{\textrm{th}} = 2.5$ ms, and $1.35$ ms for $D_{\textrm{th}} = 5$ ms. {Although the largest $D_{\textrm{th}} = 5$ ms has the longest average delay, it has the highest probability of satisfying the delay threshold.} This is because the IoT devices with critical messages have more time available to satisfy $D_{\textrm{th}}$ for larger $D_{\textrm{th}}$, and more IoT devices can learn $\beta$ with more time available. Moreover, with more time, the IoT devices with critical messages are more likely to be allocated contention-free RAP by the BS for successful RA request. Here, it is important to note that $D_{\textrm{th}}$ may not be satisfied, because $t$ and $m$ are finite. 

\begin{figure}
\centering
\includegraphics[width = 9cm]{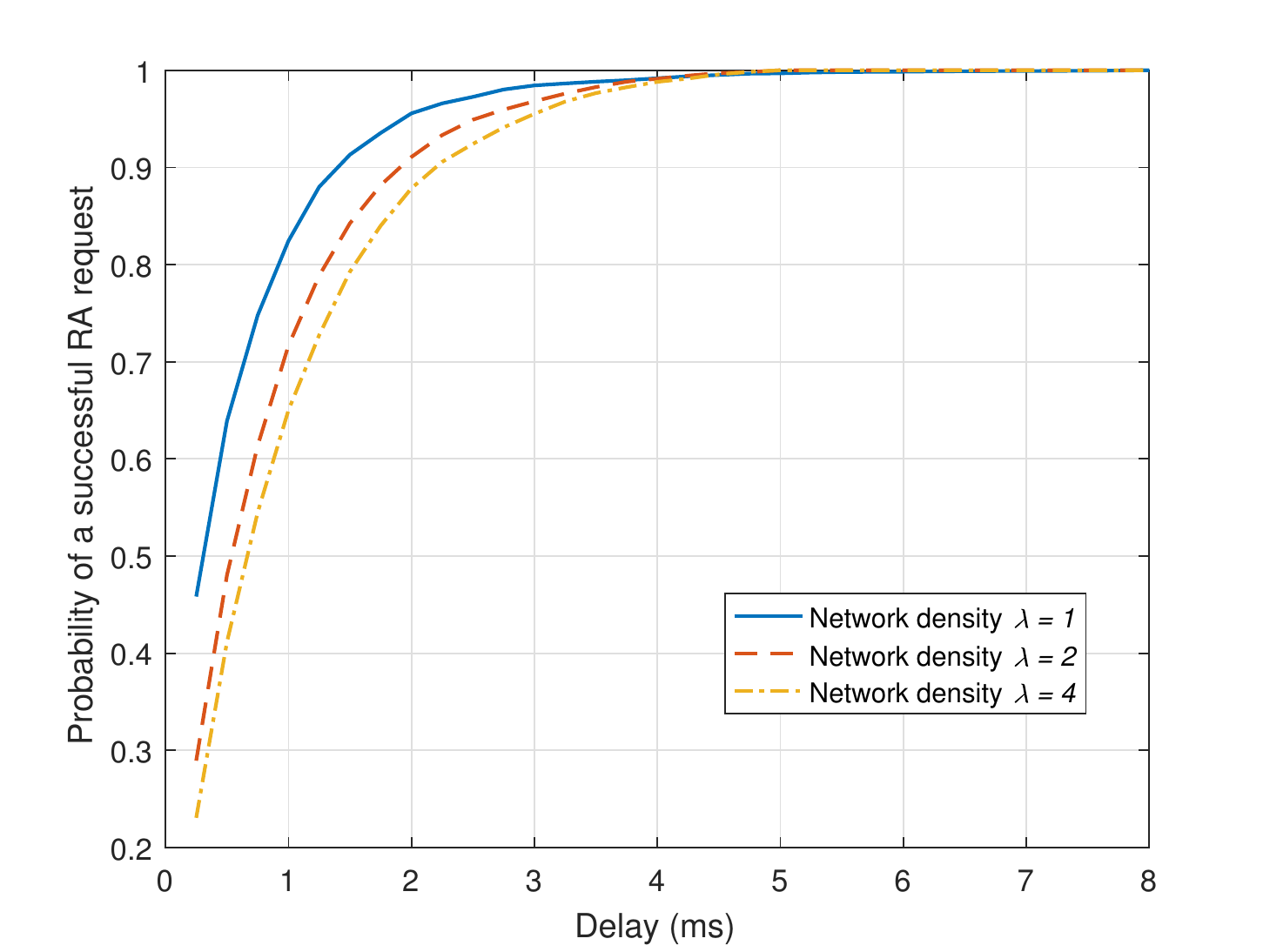}\vspace{-.2cm}
\caption{Cumulative distribution function of delay of critical messages for different $\lambda$.}\vspace{-.5cm}
\label{latenlambda}
\end{figure}

\indent Fig. \ref{latenlambda} shows the cumulative distribution function for the delay of critical messages for different network densities $\lambda$ with $D_{\textrm{th}} = 2.5$ ms, $p_{11} = 0.9$, and $m = 5$ bits. As our proposed learning method relies on the IoT devices communicating and propagating information, it is important that the IoT devices have at least one neighboring device to be able to learn $\beta$. Therefore, a massive IoT network with higher values of $\lambda$ is better for learning. However, with smaller value of $\lambda$, $T_{\textrm{min}}$ is smaller, and, thus, $D_f$ is smaller. The value of $D_f$ is $19.75$ ms for $\lambda = 1$, $39.63$ ms for $\lambda = 2$, and $79.25$ ms for $\lambda = 4$. Therefore, a small value of $\lambda$ increases the delay by increasing the number of isolated IoT devices and decreases the delay by decreasing the value of $D_f$. From Fig. \ref{latenlambda}, we can see that the probability of achieving a delay less than or equal to $D_{\textrm{th}}$ is $0.97$ for $\lambda = 1$, $0.95$ for $\lambda = 2$ and $0.92$ for $\lambda = 4$. Therefore, it is more likely to satisfy $D_{\textrm{th}}$ with lower $\lambda$, and, thus, the effect of having smaller $D_f$ is more significant than having more isolated IoT devices. Furthermore, this shows the effectiveness of our learning method in a massive network with devices that have limited resources and capabilities. 

\begin{figure}
\centering
\includegraphics[width = 9cm]{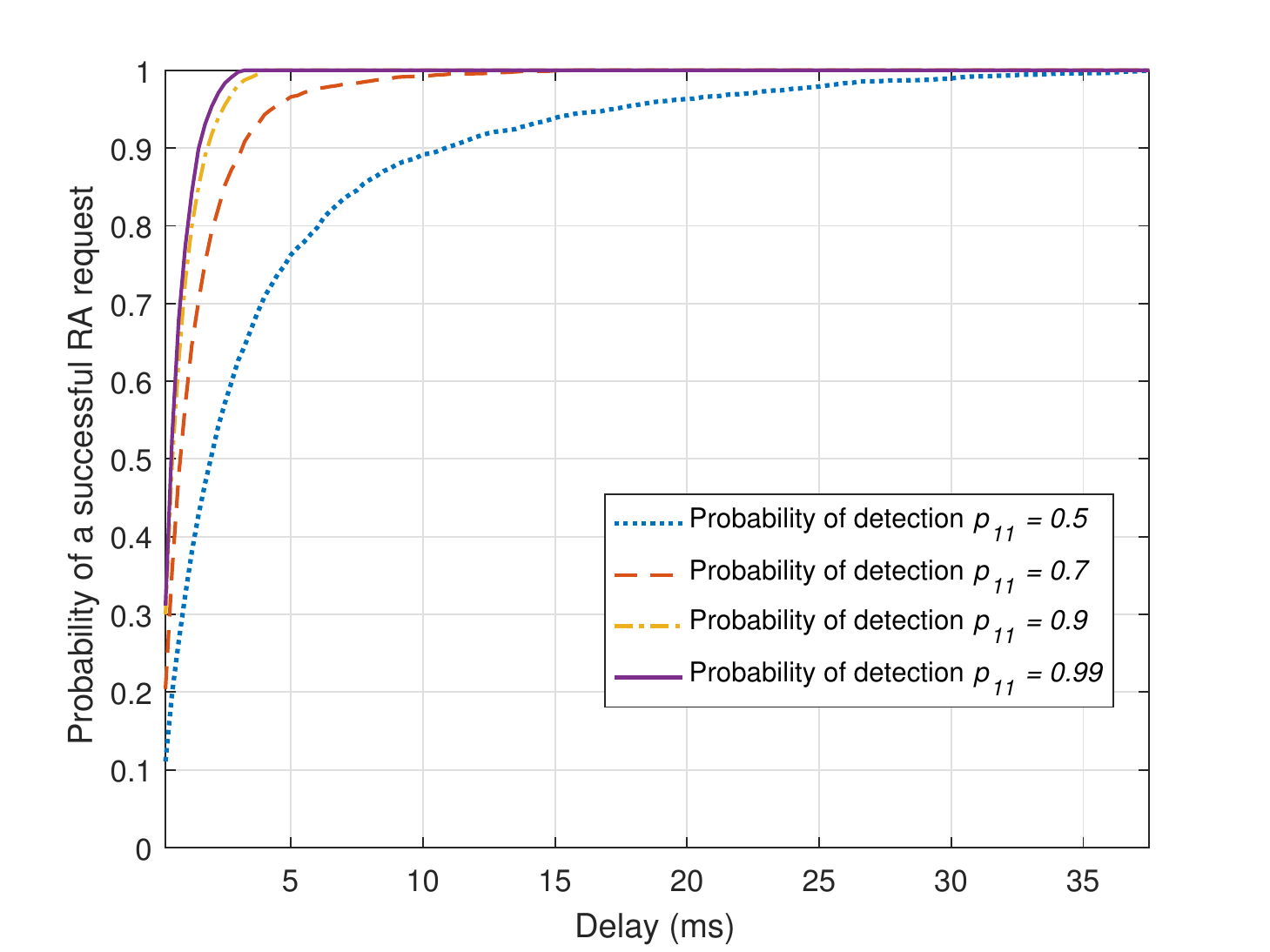}\vspace{-.2cm}
\caption{Cumulative distribution function of delay of critical messages for different $p_{11}$.}\vspace{-.5cm}
\label{latenprob}
\end{figure}

\indent Fig. \ref{latenprob} shows the cumulative distribution function for the delay of critical messages for different probabilities of detection $p_{11}$ with $D_{\textrm{th}} = 2.5$ ms, $\lambda = 2$, and $m = 5$ bits. Since the observation of an IoT device affects its learning as well as that of all IoT devices later in the sequence, the performance of our proposed method significantly depends on the quality of observations. For low $p_{11}$, the IoT devices are prone to false alarms as $p_{01} = 1 - p_{11}$, and, thus, they do not learn $N_a$ correctly, which causes higher critical message delays. The delay threshold is satisfied with probability $0.98$ for $p_{11} = 0.99$, $0.96$ for $p_{11} = 0.9$, $0.85$ for $p_{11} = 0.7$, and $0.57$ for $p_{11} = 0.5$. Unlike $D_{\textrm{th}}$ and $\lambda$, which indirectly affect the performance of learning, $p_{11}$ affects the performance of learning method directly and more significantly, because the probability of a successful RA request achieves $0.99$ after $11$ time slots. 

\begin{figure}
\centering
\includegraphics[width = 9cm]{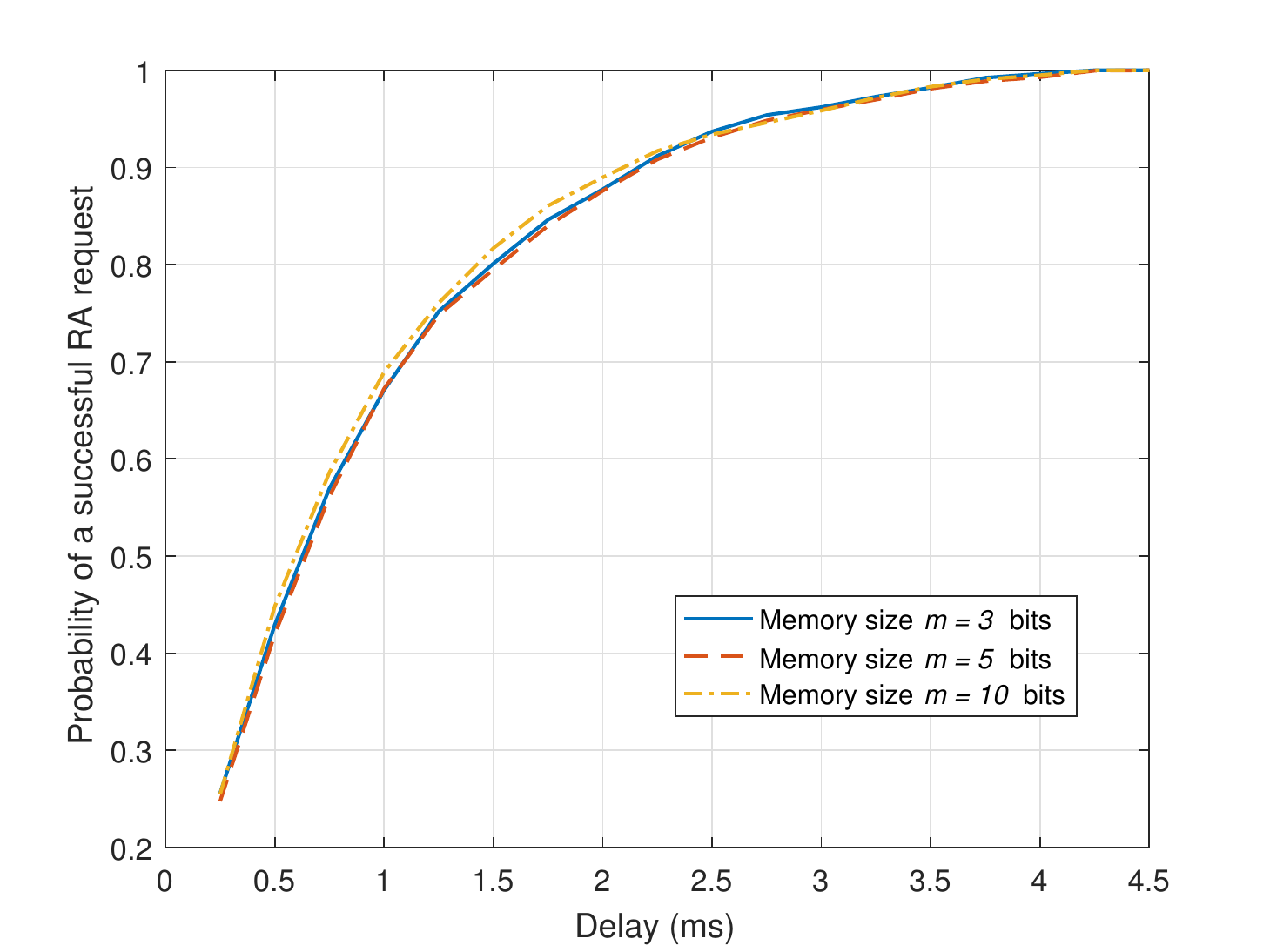}\vspace{-.2cm}
\caption{Cumulative distribution function of delay of critical messages for different $m$.}\vspace{-.5cm}
\label{latenmem}
\end{figure}

\indent Fig. \ref{latenmem} shows the cumulative distribution function for the delay of critical messages for different memory sizes $m$ with $D_{\textrm{th}} = 2.5$ ms, $\lambda = 2$, and $p_{11} = 0.9$. The effect of having larger values of $m$ is not significant in reducing the delay of critical messages even though the memory size $m$ is critical for the convergence of proposed learning method. {In our learning framework, the memory size $m$ extends the detection range $r_d$ such that the devices that are located far from the abnormal event can learn accurately. In other words, with larger $m$, the IoT devices located outside of $r_d$ have more observations of previous devices to compute more accurate private belief $x_i$ for the second phase of sequential learning. However, having large $m$ is less effective for the IoT devices within $r_d$ to learn more accurately. Since the RA request is successfully done during the first few time slots as shown in Fig. \ref{latenmem} when the devices within $r_d$ are learning, having larger $m$ has negligible effects on decreasing the delay of the critical messages. The effects of $m$ on the performance and the convergence of the learning method will be more pronounced when the devices outside of $r_d$ are learning as later shown in Fig. \ref{convmem1}.}\\
\indent The convergence of the proposed learning technique is analyzed via simulations by studying the percentage of IoT devices that learn correctly after all devices had a chance to learn $N_a$. The parameters that affect the performance of the proposed learning technique include memory size $m$, network density $\lambda$, and probability of detection $p_{11}$ within $r_d'$. The network density $\lambda$ affects the number of IoT devices that can learn such that an IoT device is more likely to have no neighboring device within $r_c$ to get necessary information for learning for lower values of $\lambda$. {Such isolated device will not be able to learn using our proposed learning method as it cannot receive necessary information for learning from its neighboring devices}. However, if the network is dense enough, the performance of learning is limited not by $\lambda$ but by memory size $m$ and probability of detection $p_{11}$. The memory size $m$ determines the effective detection radius $r_d'$ and $p_{11}$ controls how accurately the devices can observe within $r_d'$. For the subsequent simulations, we take $D_{\textrm{th}} = 3$, while varying $m$ and $p_{11}$. 

\begin{figure}[t]
\centering
\includegraphics[width = 9cm]{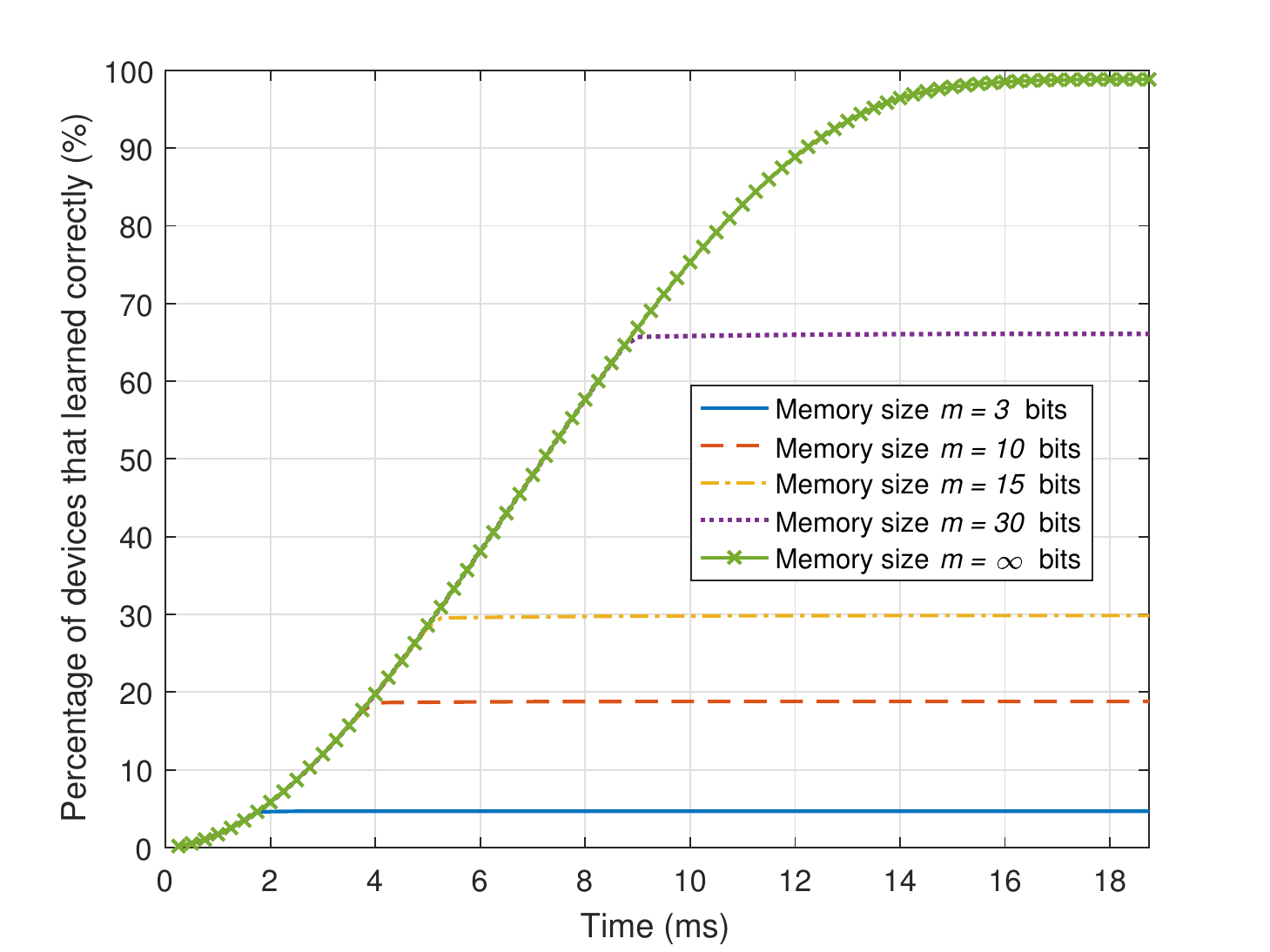}\vspace{-.2cm}
\caption{Average percentage of devices that learned correctly out of $N$ devices with varying $m$.}\vspace{-.5cm}
\label{convmem1}
\end{figure}

\indent {Fig. \ref{convmem1} shows the average percentage of devices that learned $N_a$ correctly out of all $N$ IoT devices for different memory sizes $m$ in bits with $\lambda = 2$.} As the IoT devices can learn $N_a$ correctly within $r_d'$ and $r_d'$ depends on $m$, the average percentage of devices that learned correctly is clearly limited for finite memories $m = \{3, 10, 15, 30\}$ bits, while the average percentage approaches 100\% for infinite memory $m = \infty$ bits. {Moreover, the average percentage of devices that learned correctly is an increasing function of $m$ in bits as $r_d'$ is an increasing function of $m$.} The maximum average percentage of devices that learned correctly is $4.67$\% for $m = 3$ bits, $18.79$\% for $m = 10$ bits, $29.84$\% for $m = 15$ bits, and $66.07$\% for $m = 30$ bits. This is because as the learning progresses further, the IoT devices outside of $r_d'$ will learn incorrectly, and, thus, the percentage of devices that learn correctly out of all $N$ will no longer increase. Moreover, the maximum average percentages for different values of $m$ are achieved at different time slots since a larger $m$ requires more time to achieve the maximum average percentage. {Furthermore, the results in Fig. \ref{convmem1} differ from the results shown in Fig. \ref{snapshots} in terms of the percentage of devices that learned $N_a$ correctly, because the results in Fig. \ref{snapshots} are done in smaller-scale network with different values of $w$, $l$, $r_c$, and $r_d$.}\\
\indent Fig. \ref{convprob} shows the average percentage of devices that learned $N_a$ correctly out of all $N$ IoT devices for different probabilities of detection $p_{11}$ with $m = 15$ bits and $\lambda = 2$. For lower $p_{11}$, the devices are more likely to observe and learn $N_a$ incorrectly. Learning incorrectly will affect subsequent devices in their learning as the proposed approach relies on the observation of previous IoT devices. For $p_{11} = 0.5$, less than $1\%$ of IoT devices learn correctly, and about $17.5\%$ of IoT devices learn correctly even for $p_{11} = 0.9$. Even for a $p_{11}$ close to $1$, Fig. \ref{convprob} shows that the average percentage of devices that learn $N_a$ correctly does not approach $100\%$. Hence, we can conclude that $m$ is much more critical in increasing the percentage of devices that learned correctly than $p_{11}$.

\begin{figure}[t]
\centering
\includegraphics[width = 9cm]{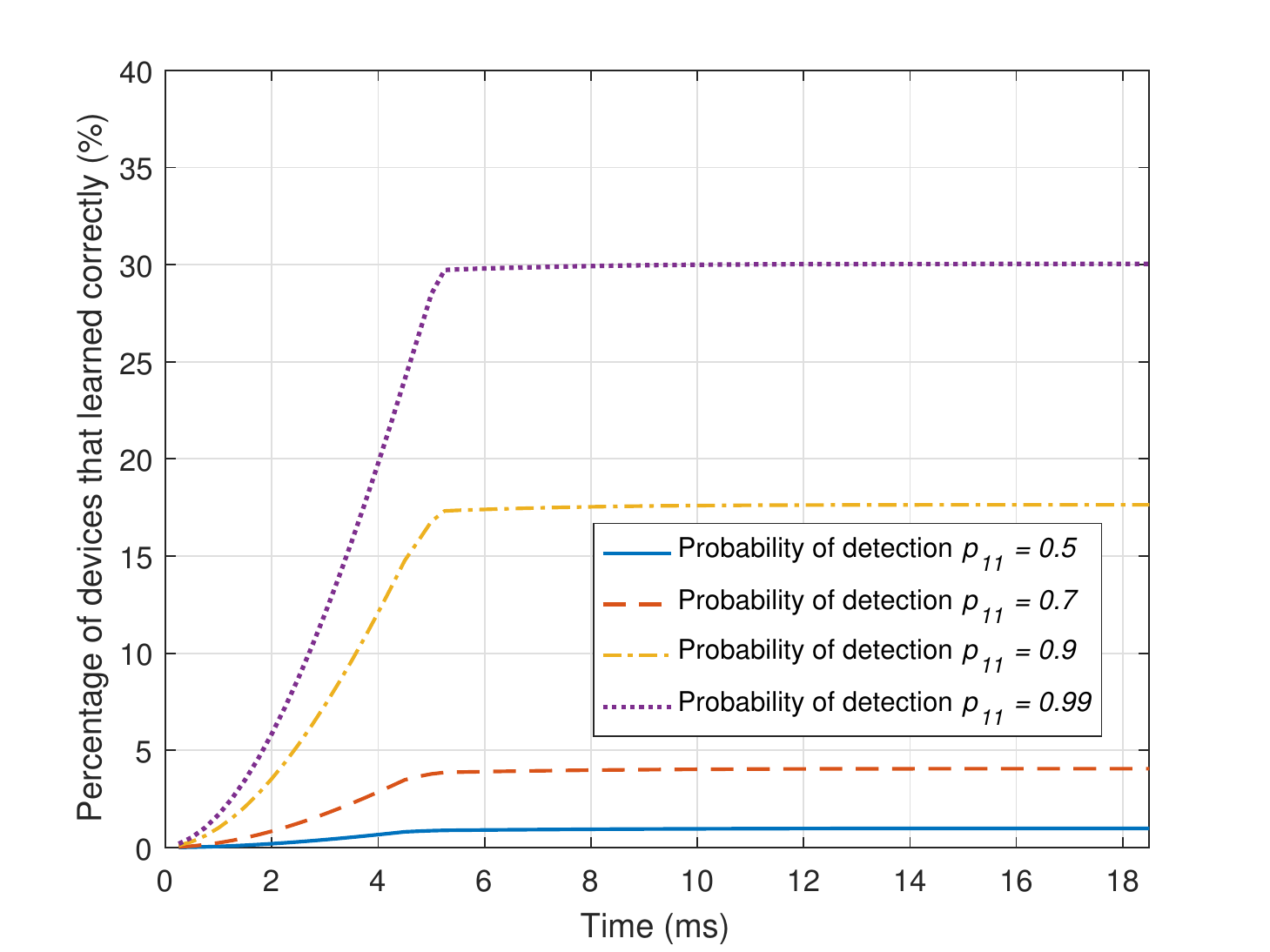}\vspace{-.2cm}
\caption{Average percentage of devices that learned correctly out of $N$ devices with varying $p_{11}$.}\vspace{-.5cm}
\label{convprob}
\end{figure}

\indent {Even with high probability of detection $p_{11}$, the average percentage of devices that learned correctly is about $30$\%. This is because the memory size $m$ is not large enough to extend $r_d'$ to circumscribe the entire deployment region.} For small $m$, some IoT devices do not satisfy \eqref{assump1} due to their limited observation range. In other words, for those devices, the effective detection range $r_d'$ is not extended enough to include the IoT devices located far away from the abnormal event, and the IoT devices located outside of $r_d'$ will not learn correctly. {Moreover, high $p_{11}$ is only applicable within the detection radius of $r_d$, and the probability of detection $p_{11}'$ outside of $r_d$ is still very low. Therefore, it is necessary for both $m$ and $p_{11}$ to be large to achieve high percentage of devices that learned correctly.} However, it is unnecessary to have all IoT devices learn correctly to satisfy a given delay threshold as shown in Fig. \ref{latendth}, Fig. \ref{latenlambda}, and Fig. \ref{latenprob}. Furthermore, since the memory size $m$ will be finite or even small in practical IoT systems, it is important to analyze the convergence of learning method with finite values of $m$ as shown in Fig. \ref{convmem1}.
\section{Conclusion}
In this paper, we have proposed a novel learning framework for enabling IoT devices to reallocate limited communication resources depending on the types and number of messages existing in the IoT. We have introduced a finite memory multi-state sequential learning using which the heterogeneous IoT devices can reallocate RAPs appropriately to reduce the delay of critical messages. We have shown that our proposed learning framework is suitable for the IoT devices with limited computational capability and finite memory and effective against the limited observation capability of IoT devices. Furthermore, we have proved the convergence of our proposed learning framework and derived the lowest expected delay of critical messages that can be achieved with our proposed learning framework. Simulation results have shown that the cumulative distribution function of delay of critical messages and average percentage of devices that learned correctly are functions of memory size, network density, and probability of detection. In particular, our proposed learning framework has shown to be effective in massive network with low delay threshold. 

\appendix \vspace{-0.2cm}
\subsection{Proof of Theorem 1} \label{app1}
The first phase of the proposed learning method is based on maximum likelihood estimation, and, thus, the chosen $s_f$ depends on the $K$ observations. Given the assumption in \eqref{assump1}, if all of the $K$ observations are $H_T$, then $s \in \mathcal{S}$ that will maximize \eqref{ml} will be $H_T$ and, thus, $s_f$ will be $H_T$. Moreover, given \eqref{likelihood1}, $\Pr(e = H_T | H_T)$ cannot be zero. Since the observations are independent, the probability of observing $K$ $H_T$ is $\Pr(e = H_T | H_T)^K$, and, thus, the probability of having $K$ observations of $H_T$ is not zero. Therefore, the favored state $s_f$ will be chosen correctly such that $s_f = H_T$ with nonzero probability from the first learning phase.\\
\indent If $s_f = H_T$ is chosen from first phase, then the true underlying state is $s_f$ and, thus, second phase will converge to $H_T$. However, if $s_f \neq H_T$ is chosen from the first phase, then the true underlying state is $s_f'$ and, thus, the second phase will not converge to $H_T$. Therefore, there is a returning condition, which is $\alpha$ consecutive observations of $s_f'$, using which the learning method can go back to the first phase if $s_f \neq H_T$ seems to be chosen from first phase.\\
\indent When $s_f \neq H_T$, the probability of observing $s_f'$ for $\alpha$ consecutive times, which may cause the learning method to go back to the first phase, is:
\begin{align}
    \hspace{-.4cm}\Pr(n_e | s_f \neq H_T) &= \Pr(\{e_{n'-\alpha+1}, \cdots, e_{n'}\}\\
    &= \{s_f', \cdots, s_f'\} \ \forall n' \leq n_e \mid s_f \neq H_T)
\end{align}
\begin{equation}
    =\left\{
                \begin{array}{ll}
                  0 \hfill &\text{if} \  n_e < \alpha,\\
                  q_1^{\alpha} \hfill &\text{if} \ n_e = \alpha,\\
                  q_1^{\alpha} + (N - \alpha)p_1 q_1^{\alpha}\\
                  \hspace{1cm} - p_1 q_1^{\alpha} \left(\sum\limits_{i=0}^{n_e - (\alpha+1)} \Pr(i | s_f \neq H_T)\right) \hfill &\text{if} \  n_e > \alpha,
                \end{array}
              \right. \label{rightn}
\end{equation}
where $q_1$ is $(1 - p_1)$ and $n_e$ is number of observations. Similarly, for the case when $s_f = H_T$, the probability of observing $s_f'$ for $\alpha$ consecutive times after $n_e$ observations is:
\begin{align}
    \hspace{-.4cm}\Pr(n_e | s_f = H_T) &= \Pr(\{e_{n'-\alpha+1}, \cdots, e_{n'}\}\\
    &= \{s_f', \cdots, s_f'\} \ \forall n' \leq n_e \mid s_f = H_T)
\end{align}
\begin{equation}
    =\left\{
                \begin{array}{ll}
                  0 \hfill &\text{if} \  n_e < \alpha,\\
                  q_2^{\alpha} \hfill &\text{if} \ n_e = \alpha,\\
                  q_2^{\alpha} + (N - \alpha)p_2 q_2^{\alpha}\\
                  \hspace{1cm} - p_2 q_2^{\alpha} \left(\sum\limits_{i=0}^{n_e - (\alpha+1)} \Pr(i | s_f = H_T)\right)\hfill &\text{if} \  n_e > \alpha,
                \end{array}
              \right. \label{wrongn}
\end{equation}
where $q_2$ is $(1 - p_2)$. Here, we note that $\Pr(n_e | s_f \neq H_T)$ and $\Pr(n_e | s_f = H_T)$ are cumulative probability distributions. Since $p_2 > p_1$, $\Pr(n_e | s_f \neq H_T)$ increases to 1 faster than $\Pr(n_e | s_f = H_T)$. However, as $n_e$ increases to infinity, $\Pr(n_e | s_f \neq H_T)$ and $\Pr(n_e | s_f = H_T)$ will both approach 1. Thus, as $n_e$ increases to infinity, the learning method will go back to the first phase regardless of $s_f$. Therefore, the probability density function $p_b$, which determines the probability of going back to Phase 1 once $s_f'$ is observed $\alpha$ consecutive times, must be designed appropriately to ensure that learning will only go back to Phase 1 when $s_f \neq H_T$.\\
\indent An appropriate choice for $p_b$ is such that $p_b$ is high for lower values of $n_e$, while $p_b$ approaches 0 as $n_e$ increases. Since $\Pr(n_e | s_f \neq H_T)$ increases to $1$ faster than $\Pr(n_e | s_f = H_T) \forall n_e$ as $p_2 > p_1$, $\Pr(n_e | s_f \neq H_T)$ is much greater than $\Pr(n_e | s_f = H_T)$ for small values of $n_e$, and, thus, it is much more likely to observe $s_f'$ for $\alpha$ consecutive times when $s_f \neq H_T$ than when $s_f = H_T$ for small values of $n_e$. Therefore, for small $n_e$, $p_b$ must be high for small values of $n_e$ so that the learning method is likely to go back to the first phase when $s_f \neq H_T$ as $\Pr(n_e | s_f \neq H_T)$ is big, while the learning method is unlikely to go back to the first phase as $\Pr(n_e | s_f = H_T)$ is small. However, as $\Pr(n_e | s_f = H_T)$ approaches 1 as $n_e$ goes to infinity, $p_b$ must approach $0$ as $n_e$ goes to infinity to prevent going back to first phase when $s_f = H_T$. One such choice for $p_b$ will be a sigmoid function that quickly decreases to 0 for $n_e > a$ for some value $a$. For instance, one such sigmoid function is:
\begin{equation}
C\left(1 - \frac{n_e - a}{1 + |n_e - a|}\right) \label{exp},
\end{equation}
where $C$ is a scaling factor to make \eqref{exp} an appropriate probability distribution. Similar to $\alpha$, the value of $a$ is a design parameter that determines the value of $n_e$ at which it becomes unlikely for the learning method to go back to the first phase. Furthermore, the value of $a$ should be chosen between the values of $n_e$ where $\Pr(n_e \mid s_f \neq H_T)$ is close to 1, while $\Pr(n_e \mid s_f = H_T)$ is close to 0, and where $\Pr(n_e \mid s_f = H_T)$ starts to significantly increase to 1.\\
\indent With appropriate choice for $p_b$, when $s_f \neq H_T$, the learning method is highly likely to go back to the first phase even when $n_e$ is low, because both $p_b$ and $\Pr(n_e \mid s_f \neq H_T)$ are high. However, when $s_f = H_T$, the learning method is highly unlikely to go back to the first phase, because either $p_b$ or $\Pr(n_e \mid s_f = H_T)$ is high, while the other is low. Therefore, as learning progresses and the number of observations $n_e$ increases to infinity, the learning method goes back to the first phase to change $s_f$ with probability approaching $0$ when $s_f = H_T$. However, $n_e$ increasing to infinity does not require the learning method to have infinite memory as our finite memory learning method updates the memory by replacing the oldest observation by newest observation. With $m$ such that $(m-2) \geq \alpha$, the finite memory will have enough observations to check if $\alpha$ consecutive observations of $s_f'$ have occurred. Since the first phase will choose $s_f = H_T$ with nonzero probability, the learning method will converge to $H_T$ in probability.

\subsection{Proof of Theorem 2} \label{app2}
For finite values of $t$ and $m$, the maximum ratio of IoT devices that learn correctly is $\frac{\pi r_t^2}{A}$, where $A$ is the area of geographical region on which the IoT devices are deployed and $r_t = \textrm{min}(tr_c, r_d')$ is a range within which the devices are likely to learn correctly. It assumes that $r_t$ is entirely within the network and all IoT devices learn correctly within $r_t$, which is best scenario for our learning method. In a given time slot at time $t$, the expected number of IoT devices with periodic messages transmitting in that time slot that learned correctly is $n_t = \frac{p_f \pi r_t^2}{A}$ and, the expected number of IoT devices with periodic messages transmitting in that time slot that did not learn correctly is $p_f - n_t$.\\
\indent Since there are IoT devices outside of $r_t$ that do not learn correctly with finite $m$ and $t$, there will be varying number of reallocated RAPs $\beta_t$ in different time slots, and the value of $\beta_t$ may range from $0$ to $\beta$. For a value of $\beta_t \in [0, \beta]$, the probability of having $\beta_t$ reallocated RAPs at any given time $t$ is:
\begin{equation}
\frac{(p_f - \beta)!}{p_f!} P(n_t, \beta_t) P(p_f - n_t, \beta - \beta_t) C(\beta, \beta_t) \label{betatprob}, 
\end{equation}
where $P(n, k)$ is $k$-permutations of $n$ and $C(n,k)$ is $k$-combinations of $n$. The probability in \eqref{betatprob} considers different cases at which $\beta_t$ IoT devices that learned correctly are assigned to use the first $\beta$ contention-free RAPs, which are to be allocated first, for their periodic messages. Furthermore, the expected number $\mathbb{E}[\beta_t]$ of contention-free RAPs that will be reallocated to contention-based RAPs is:
\begin{equation}
\mathbb{E}[\beta_t] = \frac{(p_f - \beta)!}{p_f!} \sum\limits_{b = 0}^\beta b P(n_t, b) P(p_f - n_t, \beta - b) C(\beta, b). \label{betat}
\end{equation}

\bibliographystyle{IEEEtran}

\end{document}